\documentclass[superscriptaddress, amsmath,amssymb, aps, pra, 
letterpaper, tightenlines, reprint, notitlepages]{revtex4-2}

\usepackage{graphicx}
\usepackage{floatrow}

\usepackage{dcolumn}
\usepackage{bm}

\usepackage{amsfonts}
\usepackage{dsfont}
\usepackage{bm}
\usepackage{bbold}
\usepackage{xcolor}
\usepackage{lipsum,babel}
\usepackage[normalem]{ulem}

\usepackage[utf8]{inputenc}
\usepackage{amsmath}
\usepackage{physics}
\usepackage{hyperref}
\usepackage{amssymb}
\usepackage{braket}
\usepackage{mathtools}
\usepackage{color}
\usepackage{amsthm}

\DeclareMathOperator*{\E}{\mathbb{E}}

\DeclareMathOperator*{\Haf}{\text{Haf}}

\newcommand{\cor}[1]{{\color{red}{#1}}}

\theoremstyle{plain}
\newtheorem{theorem}{Theorem}
\newtheorem{corollary}{Corollary}

\newtheorem{lemma}{Lemma}
\newtheorem{conjecture}{Conjecture}

\theoremstyle{definition}
\newtheorem{problem}{Problem}

\begin{document}
	
\definecolor{red}{RGB}{255,0,0}
\preprint{APS/123-QED}

\title{Sufficient conditions for hardness of lossy Gaussian boson sampling}

\author{Byeongseon Go}
\affiliation{NextQuantum Innovation Research Center, Department of Physics and Astronomy, Seoul National University, Seoul 08826, Republic of Korea}
\author{Changhun Oh}
\email{changhun0218@gmail.com}
\affiliation{Department of Physics, Korea Advanced Institute of Science and Technology, Daejeon 34141, Republic of Korea}
\author{Hyunseok Jeong}
\email{h.jeong37@gmail.com}
\affiliation{NextQuantum Innovation Research Center, Department of Physics and Astronomy, Seoul National University, Seoul 08826, Republic of Korea}

\begin{abstract}

Gaussian boson sampling (GBS) is a prominent candidate for the experimental demonstration of quantum advantage. 
However, while the current implementations of GBS are unavoidably subject to noise, the robustness of the classical intractability of GBS against noise remains largely unexplored.
In this work, we establish the complexity-theoretic foundations for the classical intractability of noisy GBS under photon loss, which is a dominant source of imperfection in current implementations. 
We identify the loss threshold below which lossy GBS maintains the same complexity-theoretic level as ideal GBS, and show that this holds when at most a logarithmic fraction of photons is lost.
We additionally derive an intractability criterion for the loss rate through a direct quantification of the statistical distance between ideal and lossy GBS.
This work presents the first rigorous characterization of classically intractable regimes of lossy GBS, thereby serving as a crucial step toward demonstrating quantum advantage with near-term implementations.

\end{abstract}

\maketitle

\emph{Introduction.---}
Gaussian boson sampling (GBS) is the task of sampling from the output distribution obtained by photon-number measurements on Gaussian states~\cite{hamilton2017gaussian, kruse2019detailed, deshpande2022quantum, grier2022complexity, lund2014boson}. 
Owing to its strong guarantees of classical intractability and experimental feasibility, GBS has become a prominent candidate for the near-term demonstration of quantum advantage, motivating substantial experimental progress to date~\cite{zhong2020quantum, zhong2021phase, madsen2022quantum, deng2023gaussian, liu2025robust}.
However, the presence of imperfections in experiments still remains a major challenge.
Physical imperfections can significantly reduce the computational complexity of boson sampling~\cite{qi2020regimes, villalonga2021efficient, bulmer2022boundary, liu2023simulating, oh2024classical, umanskii2024classical, shi2022effect, bulmer2024simulating, rahimi2016sufficient, renema2018efficient,  moylett2019classically, shchesnovich2019noise,  garcia2019simulating, oszmaniec2018classical,  brod2020classical,  oh2021classical, oh2023classical, oh2025classical, renema2018classical, renema2020simulability, van2024efficient}, and numerous classical algorithms have been developed to simulate GBS under such imperfections~\cite{qi2020regimes, villalonga2021efficient, bulmer2022boundary, liu2023simulating, oh2024classical, umanskii2024classical, shi2022effect, bulmer2024simulating, rahimi2016sufficient}. 
These results imply that increasing the noise rate reduces the complexity of GBS, eventually rendering it classically simulable beyond a certain threshold.
Namely, this threshold identifies the regime one must \textit{avoid} in order to achieve quantum advantage with GBS experiments, thereby providing a necessary but not sufficient condition for such demonstrations.

Taking it a step further, to provide a sufficient condition, it is crucial to identify the classically \textit{intractable} noise regime of GBS.
Specifically, as the noise rate in GBS decreases, one expects a noise threshold below which noisy GBS retains the hardness of ideal GBS.
However, in contrast to the numerous analyses on classical simulability of noisy GBS~\cite{qi2020regimes, villalonga2021efficient, bulmer2022boundary, liu2023simulating, oh2024classical, umanskii2024classical, shi2022effect, bulmer2024simulating, rahimi2016sufficient}, this opposite direction remains largely unknown.
Namely, we have limited knowledge of the boundary of the noise regime one must \textit{attain} in order to achieve quantum advantage.
In particular, establishing hardness result under \textit{photon loss} is especially crucial, since it is a leading cause of error in current GBS experiments~\cite{zhong2020quantum, zhong2021phase, madsen2022quantum, deng2023gaussian, liu2025robust}, and also a key ingredient enabling efficient classical simulation of noisy GBS~\cite{qi2020regimes, villalonga2021efficient, bulmer2022boundary, liu2023simulating, oh2024classical, umanskii2024classical}, thereby constituting the main obstacle to demonstrating quantum advantage experimentally.

In this work, we establish complexity-theoretic foundations on the classical intractability of GBS under photon loss, hereafter referred to as \textit{lossy GBS}. 
We characterize the loss threshold below which lossy GBS retains the same complexity-theoretic level as ideal GBS, showing that lossy GBS maintains the hardness of ideal GBS when the average number of lost photons is at most \textit{logarithmic} to the input photons.
Additionally, we present another loss threshold derived from information-theoretic bounds, offering an alternative approach to derive the hardness of lossy GBS. 
To the best of our knowledge, this work is the first to establish complexity-theoretic foundations that rigorously characterize the classically intractable regimes of lossy GBS.
Since recent quantum advantage experiments are largely GBS-based~\cite{zhong2020quantum, zhong2021phase, madsen2022quantum, deng2023gaussian, liu2025robust} and are primarily limited by photon loss~\cite{qi2020regimes, villalonga2021efficient, bulmer2022boundary, liu2023simulating, oh2024classical, umanskii2024classical}, we expect that our findings provide a crucial theoretical foundation for realizing quantum advantage under realistic experimental conditions.

\emph{Our settings.---}
We first introduce our lossy GBS setup, which is illustrated in Fig.~\ref{fig: schematics}(a).
We consider $M$ squeezed vacuum states with identical squeezing parameter $r$ as input, with $M$ being even for simplicity. 
We inject them into an $M$-mode random linear optical circuit, characterized by an $M$ by $M$ Haar-random unitary matrix $U$.
To describe photon loss, we adopt the beam splitter loss model~\cite{oszmaniec2018classical, garcia2019simulating, qi2020regimes, oh2024classical}, where the degree of loss is characterized by a uniform \textit{transmission rate} $\eta \in [0,1]$ (see Fig.~\ref{fig: schematics}(b)).
At the end, all modes are measured with photon-number detectors, where we consider \textit{post-selecting} $N$-photon outcomes with even $N$ for convenience.
The $N$-photon outcome can be represented as an $N$-dimensional vector $S = (s_1,\dots,s_N)$ with $1\leq s_i\leq M$ and $s_1\le s_2\le\cdots\le s_N$, where $s_i$ denotes the mode in which the $i$th photon is detected.

\begin{figure}[t]
\includegraphics[width=\linewidth]{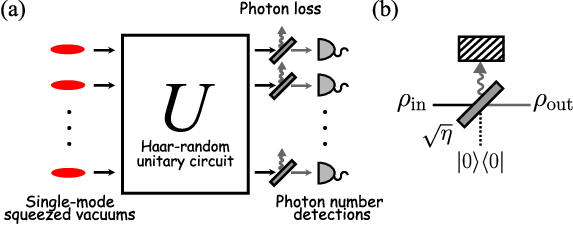}
\caption{(a) Schematic of our lossy GBS setup, composed of input $M$ squeezed vacuum states, an $M$-mode random linear optical circuit, beam splitter loss described in (b), and measurement of $N$ output photons.
(b) The beamsplitter loss model~\cite{oszmaniec2018classical, garcia2019simulating, qi2020regimes, oh2024classical, barnett1998quantum, demkowicz2015quantum}.
Each input mode interacts with an ancillary vacuum mode via a beamsplitter of transmittance $\sqrt{\eta}$, after which the ancillary mode is traced out, as indicated by the hatched square in the figure.
}
\label{fig: schematics}
\end{figure}

In this setting, we find that the output probability $p_{S}(\eta, U)$ of lossy GBS can be expressed as 
\begin{align}\label{noisypostselectedoutputprobability}
\begin{split}
p_{S}(\eta, U) &= \frac{1}{\mu(S)}\binom{\frac{M}{2} + \frac{N}{2} -1 }{\frac{N}{2}}^{-1}Q(\eta)^{-1} \\
&\times \Haf{\begin{pmatrix}
(UU^{T})_{S} & (1-\eta)\tanh{r}\mathbb{I}_{N} \\
(1-\eta)\tanh{r}\mathbb{I}_{N} & (U^{*}U^{\dag})_{S} 
\end{pmatrix}}    ,
\end{split}
\end{align}
where $\mu(S)$ denotes the product of the multiplicities of all entries in $S$ (e.g., $\mu(S) = 1$ for a collision-free $S$), and $(UU^T)_{S}$ is the $N$ by $N$ matrix obtained by taking rows and columns of $UU^T$ according to $S$  (similarly for $(U^*U^{\dag})_S$). 
Also, $\mathbb{I}_{N}$ is the $N$ by $N$ identity matrix, and the multiplicative factor $Q(\eta)$ in the right-hand side of Eq.~\eqref{noisypostselectedoutputprobability} is defined as 
\begin{align}
\begin{split}
Q(\eta) &\coloneqq (1-(1-\eta)^2\tanh^2{r})^{\frac{M}{2}+N} \\
&\times {}_2F_1\left(\frac{M + N}{2}, \frac{N + 1}{2}; \frac{1}{2};(1-\eta)^2\tanh^2r\right),
\end{split}
\end{align}
where ${}_2F_1(a,b;c;z)$ denotes the hypergeometric function~\cite{whittaker1920course}. 
We remark that when $\eta = 1$, implying that no loss has occurred, $p_{S}(\eta, U)$ in Eq.~\eqref{noisypostselectedoutputprobability} is identical to the output probability of (post-selected) ideal GBS~\cite{hamilton2017gaussian, kruse2019detailed, deshpande2022quantum}. 
A detailed derivation of $p_{S}(\eta, U)$ in Eq.~\eqref{noisypostselectedoutputprobability} is provided in Sec.~\ref{Appendix:section:outputprobabilityoflossyGBS} of the Supplemental Material.

To further specify our settings, let $M \propto N^{\gamma}$ with large $\gamma > 2$ to ensure the hiding property of GBS~\cite{jiang2009entries, shou2025proof} and the dominance of collision-free outcomes (by bosonic birthday paradox~\cite{aaronson2011computational, arkhipov2012bosonic, qi2020regimes, deshpande2022quantum}), both of which are crucial for the hardness analysis.
Also, following prior analyses of ideal GBS~\cite{hamilton2017gaussian, kruse2019detailed}, where the input squeezing is adjusted to set $N$ as the mean number of output photons, we adjust the input squeezing $r$ to satisfy $N = \eta M \sinh^2 r$ for given $N$, $M$, and $\eta$.
This ensures that the post-selection probability of observing $N$ photons satisfies $\Pr[N] \geq \Omega(\frac{1}{\sqrt{N}})$ (for details, see Sec.~\ref{Appendix:section:postselectionprobabilityoflossyGBS} in the Supplemental Material). 
Hence, as argued in~\cite{aaronson2016bosonsampling, hamilton2017gaussian, kruse2019detailed}, our hardness result on post-selected lossy GBS can carry over to the hardness of \textit{general} lossy GBS without post-selection.

\emph{Classical hardness of lossy GBS.---}
We now present our hardness result for lossy GBS, which is outlined in Fig.~\ref{fig: summary}.
To summarize our result, we show that when the transmission rate $\eta^*$ of lossy GBS satisfies $\eta^* \geq \eta_{\rm{th}}$ for a threshold $\eta_{\rm{th}}$, such lossy GBS is classically hard under certain complexity-theoretic assumptions.

\begin{figure}[b]
\includegraphics[width=0.85\linewidth]{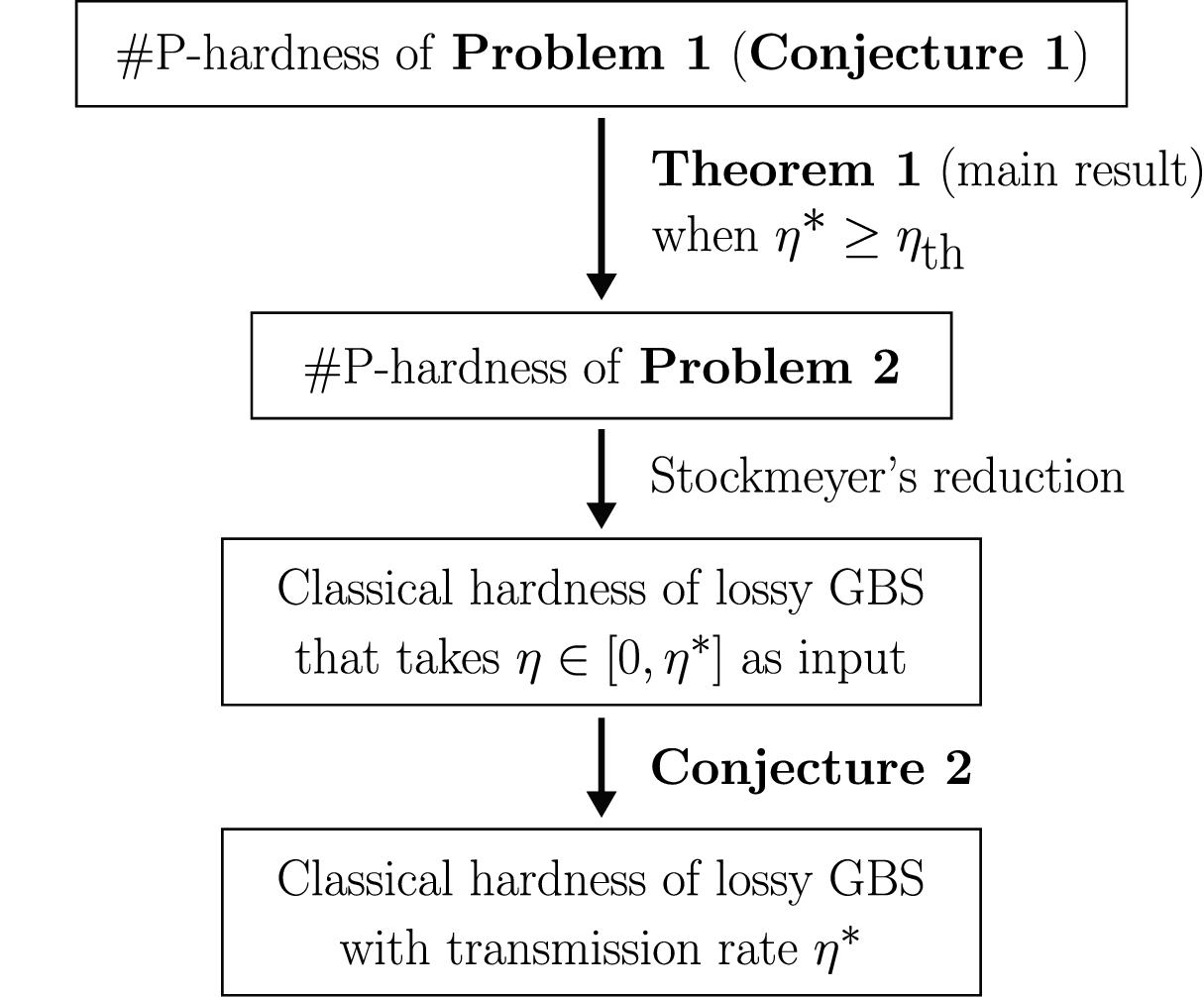}
\caption{Outline of our hardness analysis on lossy GBS.
To summarize, we show that when the transmission rate $\eta^*$ of lossy GBS exceeds a threshold $\eta_{\rm{th}}$ given in Theorem~\ref{thm: main result informal}, such lossy GBS is classically hard to simulate under certain complexity-theoretic assumptions. 
}
\label{fig: summary}
\end{figure}

We first sketch how to establish the classical hardness of sampling problems, i.e., approximate sampling to within total variation distance error.
The state-of-the-art proof technique is built upon the (conjectured) \#P-hardness of estimating average-case output probabilities~\cite{aaronson2011computational, hamilton2017gaussian, kruse2019detailed, bouland2019complexity, grier2022complexity, deshpande2022quantum, bremner2016average, movassagh2023hardness, kondo2022quantum, bouland2022noise, bouland2023complexity, bouland2024average, go2024exploring, go2024computational}. 
Specifically, given a polynomial-time classical algorithm for the approximate sampling, Stockmeyer's algorithm~\cite{stockmeyer1985approximation} enables one to \textit{well-estimate} average-case output probabilities over randomly chosen circuits in complexity $\rm{BPP}^{\rm NP}$.
Hence, by Stockmeyer's reduction, if this probability estimation task is proven to be \#P-hard, it follows that approximate sampling must be classically hard, unless the polynomial hierarchy collapses to $\rm{BPP}^{\rm NP}$.

For the classical hardness of ideal GBS, let us denote $p_{S}(U) \coloneqq p_{S}(1, U)$ as the ideal output probability. 
Then, the probability estimation task underlying the hardness of ideal GBS can be formalized as~\cite{aaronson2011computational, hamilton2017gaussian, kruse2019detailed, deshpande2022quantum, bouland2022noise, bouland2023complexity, bouland2024average}

\begin{problem}\label{problem: average-case for ideal GBS}
    On input a Haar-random unitary matrix $U$ and error bounds $\epsilon_0,\, \delta_0 > 0$, outputs an estimate of $p_{S}(U)$ for a collision-free $S$, to within $\pm \epsilon_0 \binom{M}{N}^{-1}$ over $1-\delta_0$ fraction of $U$.
\end{problem}

In line with previous hardness arguments on ideal GBS~\cite{hamilton2017gaussian, kruse2019detailed}, we conjecture that Problem~\ref{problem: average-case for ideal GBS} is \#P-hard.

\begin{conjecture}[\cite{hamilton2017gaussian, kruse2019detailed}]\label{conjecture: sharp-p-hard}
    Problem~\ref{problem: average-case for ideal GBS} is \#$\rm{P}$-hard.
\end{conjecture}

Under Conjecture~\ref{conjecture: sharp-p-hard}, one can establish the hardness of simulating ideal GBS; further details are provided in Sec.~\ref{supplement: hardness proof of GBS} of the Supplemental Material.







%
For the hardness of lossy GBS, we now extend this argument to lossy GBS.
We similarly define the computational task of estimating the output probability $p_{S}(\eta, U)$ for lossy GBS.

\begin{problem}\label{problem: average-case for lossy GBS}
    On input a Haar-random unitary matrix $U$, a transmission rate $\eta \in [0, \eta^*]$ with a fixed $\eta^*$, and error bounds $\epsilon,\,\delta > 0$, outputs an estimate of $p_{S}(\eta, U)$ for a collision-free $S$, to within $\pm \epsilon \binom{M}{N}^{-1}$ over $1-\delta$ fraction of $U$.
\end{problem}

Notice that $\eta$ is now treated as an input parameter to Problem~\ref{problem: average-case for lossy GBS}, while $\eta^*$ denotes a fixed transmission rate characterizing the noise level of lossy GBS for which we establish the hardness result.

Then, our main result states that Problem~\ref{problem: average-case for lossy GBS} is \textit{as hard as} Problem~\ref{problem: average-case for ideal GBS} whenever $\eta^*$ exceeds a certain threshold.

\begin{theorem}\label{thm: main result informal}
    There exists a threshold transmission rate $\eta_{\rm{th}}$ satisfying
    \begin{align}\label{eq: noise threshold for hardness}
        (1-\eta_{\rm{th}})N = O(\log N),
    \end{align}
    such that when $\eta^* \geq \eta_{\rm{th}}$ and $\epsilon, \delta = {\rm{poly}}(N, \epsilon_0^{-1}, \delta_0^{-1})^{-1}$, Problem~\ref{problem: average-case for lossy GBS} is at least as hard as Problem~\ref{problem: average-case for ideal GBS}.
\end{theorem}

To prove Theorem~\ref{thm: main result informal}, we establish a complexity-theoretic reduction from Problem~\ref{problem: average-case for ideal GBS} to Problem~\ref{problem: average-case for lossy GBS}, similarly to the hardness proof of noisy Fock state boson sampling described in~\cite{aaronson2016bosonsampling, go2025quantum}.
Informally, we show that one can well-estimate the ideal value $p_{S}(U) = p_{S}(1, U)$ if one can well-estimate $p_{S}(\eta, U)$ for an arbitrarily given $\eta \in [0,\eta^*]$. 
Again, we note that the transmission rate $\eta \in [0, \eta^*]$ is treated as an input variable to Problem~\ref{problem: average-case for lossy GBS}.
This later plays a crucial role in proving Theorem~\ref{thm: main result informal}, as our reduction process involves inferring $p_{S}(U)$ from multiple values of $p_{S}(\eta, U)$, each corresponding to a different transmission rate $\eta$.

Before outlining the proof of Theorem~\ref{thm: main result informal}, we give a step-by-step explanation of how the theorem yields the classical hardness of lossy GBS, when its transmission rate $\eta^*$ exceeds $\eta_{\rm{th}}$ in Eq.~\eqref{eq: noise threshold for hardness}.
First, under Conjecture~\ref{conjecture: sharp-p-hard}, Theorem~\ref{thm: main result informal} establishes the \#P-hardness of Problem~\ref{problem: average-case for lossy GBS}.
Then, by directly following the previous hardness proofs based on Stockmeyer's reduction~\cite{aaronson2011computational, hamilton2017gaussian, kruse2019detailed, deshpande2022quantum, bouland2022noise, bouland2023complexity, bouland2024average}, this further implies the classical hardness of simulating lossy GBS that can choose \textit{any} transmission rate $\eta \in [0, \eta^*]$.
In other words, when $\eta^* \geq \eta_{\rm{th}}$, there exists at least one $\eta \in [0, \eta^*]$ such that simulating lossy GBS with transmission rate $\eta$ is classically hard.

However, this alone does not suffice to conclude that lossy GBS with the specific transmission rate $\eta^*$ is classically hard; an additional assumption is required.
Based on numerous observations~\cite{renema2018efficient,  moylett2019classically, shchesnovich2019noise, garcia2019simulating, oszmaniec2018classical,  brod2020classical,  oh2021classical, oh2023classical, oh2025classical, noh2020efficient, aharonov2023polynomial, gao2018efficient, bremner2017achieving, rajakumar2025polynomial, aharonov1996limitations, lee2025classical, muller2016relative, nelson2025limitations, oh2025recent}, increasing noise tends to make quantum systems more classically simulable, and lossy GBS likewise becomes easier to simulate as the transmission rate $\eta$ decreases~\cite{qi2020regimes, villalonga2021efficient, bulmer2022boundary, liu2023simulating, oh2024classical, umanskii2024classical}.
These observations naturally motivate the following physically plausible conjecture.

\begin{conjecture}[Loss reduces complexity of GBS]\label{conjecture: complexity decrease}
    Let the loss rate in lossy GBS be characterized by the uniform transmission rate $\eta \in [0,1]$. 
    Then, the computational complexity of simulating lossy GBS to within a fixed total variation distance decreases monotonically as $\eta$ decreases.
\end{conjecture}

Then, under Conjecture~\ref{conjecture: complexity decrease}, if there exists $\eta \in [0, \eta^*]$ for which simulating lossy GBS is classically hard, it follows that simulating lossy GBS with the specific transmission rate $\eta^*$ is also classically hard.


By collecting all the preceding arguments, we finally arrive at the following corollary.

\begin{corollary}[Hardness of lossy GBS]
    Lossy GBS with transmission rate $\eta^*$ exceeding the threshold $\eta_{\rm{th}}$ in Eq.~\eqref{eq: noise threshold for hardness} is classically hard under Conjecture~\ref{conjecture: sharp-p-hard} and Conjecture~\ref{conjecture: complexity decrease}.
    In other words, such a lossy GBS preserves the classical hardness of ideal GBS under Conjecture~\ref{conjecture: complexity decrease},
\end{corollary}

We lastly highlight the implication of our hardness criterion in Theorem~\ref{thm: main result informal}.
Specifically, consider a general lossy GBS (i.e., without post-selection) with its transmission rate $\eta^* \geq \eta_{\rm{th}}$ and mean output photon number $N = \eta^* M \sinh^2 r$, for which our result ultimately implies the classical hardness.
Then, the mean number of lost photons $(\frac{1}{\eta^*} - 1)N$ is at most logarithmically related to the mean number of input photons $\frac{1}{\eta^*}N$. 
Therefore, our result implies the classical hardness of lossy GBS when at most a \textit{logarithmic} number of photons are lost on average.
This scaling of tolerable noisy photons for the hardness matches the scaling obtained for the hardness of Fock state boson sampling under partial distinguishability noise~\cite{go2025quantum}.

\begin{proof}[Proof Sketch of Theorem~\ref{thm: main result informal}]

We now describe how the ideal output probability $p_{S}(U) \coloneqq p_{S}(1, U)$ in Problem~\ref{problem: average-case for ideal GBS} can be inferred using the estimated values of lossy output probability $p_{S}(\eta, U)$ in Problem~\ref{problem: average-case for lossy GBS} for $\eta \in [0, \eta^*]$.

First, by the hiding property of GBS~\cite{jiang2009entries, shou2025proof}, for $M$-dimensional Haar-random unitary matrix $U$ and a collision-free outcome $S$, we have $(UU^{T})_{S} \simeq \frac{1}{M}XX^{T}$ for $N$ by $M$ i.i.d. Gaussian matrix $X \sim \mathcal{N}(0,1)_{\mathbb{C}}^{N \times M}$. 
Hence, by rescaling both $p_{S}(U)$ and $p_{S}(\eta, U)$ with the same multiplicative factors, the problem reduces to estimating the ideal value
\begin{align}\label{eq: rescaled ideal probability}
   P(X)  = \Haf{\begin{pmatrix}
XX^T & 0 \\
0 & X^*X^{\dag} 
\end{pmatrix}}
\end{align}
to within additive imprecision $\epsilon_0 \binom{M}{N}^{-1} M^N     \binom{\frac{M}{2} + \frac{N}{2} - 1}{\frac{N}{2}} \sim \epsilon_0\binom{\frac{M}{2} + \frac{N}{2} - 1}{\frac{N}{2}}N!$ over $1-\delta_0$ fraction of $X$, using the estimates
\begin{align}
&P(\eta, X) \nonumber \\ 
&=  Q(\eta)^{-1}  \Haf{\begin{pmatrix}
XX^T & (1-\eta)M\tanh{r}\mathbb{I}_{N} \\
(1-\eta)M\tanh{r}\mathbb{I}_{N} & X^*X^{\dag}  
\end{pmatrix}} \label{qqqe} \\
&= Q(\eta)^{-1}R_{X}(\eta) \label{qqe}
\end{align}
within additive imprecision $\epsilon\binom{\frac{M}{2} + \frac{N}{2} - 1}{\frac{N}{2}}N!$ over $1-\delta$ fraction of $X$, for $\eta \in [0, \eta^*]$.
Here, $R_{X}(\eta)$ denotes the hafnian term in the right-hand side of Eq.~\eqref{qqqe}, which gives the desired value $R_{X}(\eta) = P(X)$ at $\eta = 1$.

To proceed, we find that the multiplicative factor $Q(\eta)$ in the right-hand side of Eq.~\eqref{qqe} is efficiently computable to within a desired accuracy.
This allows one to estimate $R_{X}(\eta)$ for any $\eta \in [0, \eta^*]$ and over $1 - \delta$ fraction of $X$, by multiplying the computed $Q(\eta)$ to the estimated values of $P(\eta, X)$.
Moreover, we also find that $R_{X}(\eta)$ can be expressed as an $N$-degree polynomial in $\eta$ such that
\begin{align}\label{eq: series expansion}
     R_{X}(\eta) = c_0(1-\eta)^0 + c_2(1-\eta)^2 + \dots + c_N (1 - \eta)^{N} ,
\end{align}
where the zeroth-order term $c_0$ equals the ideal value $P(X)$.


Based on this understanding, our key idea is to use polynomial interpolation to infer the value at $\eta = 1$.
Here, since directly using the $N$-degree polynomial $R_{X}(\eta)$ for the polynomial interpolation would induce an exponential imprecision blowup of $e^{O(N)}$, we instead employ a \textit{low-degree} polynomial $R_{X}^{(l)}(\eta)$ with its degree $l$, obtained by truncating high-order terms beyond $(1-\eta)^{l}$ in $R_{X}(\eta)$.
Importantly, provided that $l \gg (1-\eta^*)N$, the low-degree polynomial $R_{X}^{(l)}(\eta)$ can be made arbitrarily close to $R_{X}(\eta)$ over a large fraction of $X \sim \mathcal{N}(0,1)_{\mathbb{C}}^{N \times M}$, for certain values of $\eta \in [0,\eta^*]$. 
Furthermore, as the truncation preserves the zeroth-order term of $R_{X}(\eta)$ (i.e., $c_0$ in Eq.~\eqref{eq: series expansion}), $R_{X}^{(l)}(\eta)$ still retains the ideal value $P(X)$ at $\eta = 1$.

\begin{figure}[t]
\includegraphics[width=0.95\linewidth]{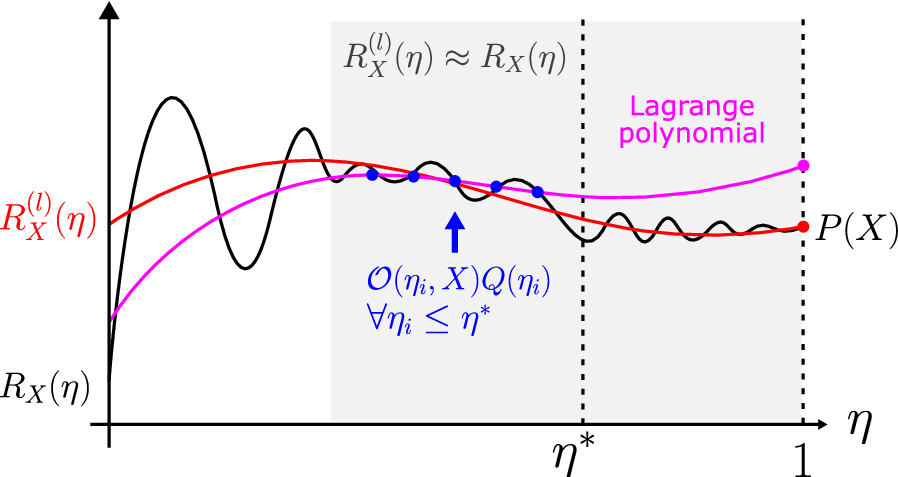}
\caption{Schematics for the proof of Theorem~\ref{thm: main result informal}. 
Given an oracle $\mathcal{O}$ that estimates $P(\eta,X)$ in Eq.~\eqref{qqe} for $\eta \leq \eta^*$ and $X \sim \mathcal{N}(0,1)_{\mathbb{C}}^{N \times M}$, one can estimate the $N$-degree polynomial $R_{X}(\eta)$ (black curve) through $\mathcal{O}(\eta, X)Q(\eta)$ (blue dots) by explicitly computing $Q(\eta)$.
This also allows estimation of the low-degree polynomial $R_X^{(l)}(\eta)$ (red curve) for $\eta$ satisfying $R_{X}^{(l)}(\eta) \approx R_{X}(\eta)$ (grey area) and $\eta \leq \eta^*$.
Using the estimated values of $R_{X}^{(l)}(\eta)$ for different $\eta$ values, one can construct the corresponding Lagrange interpolation polynomial (magenta curve), which in turn allows estimating the value $R_{X}^{(l)}(1) = P(X)$ in Eq.~\eqref{eq: rescaled ideal probability}.
}
\label{fig: proofsketch}
\end{figure}

In short, $P(X)$ can be inferred via performing polynomial interpolation of the low-degree polynomial $R_{X}^{(l)}(\eta)$, incurring an imprecision blowup of $e^{O(l)}$, from successfully estimated values of $P(\eta, X)$ (i.e., succeed over a large fraction of $X$) for certain values of $\eta \in [0, \eta^*]$.
Given that $\eta^* \geq \eta_{\rm{th}}$ for $\eta_{\rm{th}}$ satisfying Eq.~\eqref{eq: noise threshold for hardness}, the degree $l$ can be set to $O(\log \text{poly}(N, \epsilon_0^{-1}, \delta_0^{-1}))$ while satisfying $l \gg (1-\eta^*)N$.
This ensures that the required error bounds $\epsilon,\,\delta$ can be made polynomially related to $N^{-1}$, $\epsilon_0$, and $\delta_0$, thereby enabling the polynomial reduction.
The overall proof process is illustrated in Fig.~\ref{fig: proofsketch}.
For a detailed analysis, see Sec.~\ref{supplement: main result} in the Supplemental Material.





\end{proof}

\emph{Another hardness criterion for lossy GBS.---}
In addition to our main result, we derive a loss-rate threshold for the classical hardness of lossy GBS through an alternative approach.
Although the bound derived below yields a looser loss-rate threshold for the hardness compared with our main result in Theorem~\ref{thm: main result informal}, it still allows for further improvement and thus could offer an alternative path toward stronger hardness results of lossy GBS.

Previously, we obtained the loss-rate threshold by establishing a complexity-theoretic reduction, deriving the average-case hardness of output probability estimation for lossy GBS from that of ideal GBS. 
Instead, one would also consider a direct quantification of the statistical distance between ideal and lossy GBS using information-theoretic bounds.  
In fact, the effect of photon loss in ideal GBS can be characterized by the total variation distance from the ideal distribution.
Accordingly, the total variation distance between the output probability distribution of ideal and lossy GBS can be bounded in terms of transmission rate $\eta$, as formulated in the following lemma.


\begin{lemma}\label{theorem: total variation distance bound}
    The total variation distance between the output probability distribution of ideal and the lossy GBS is bounded by $\beta < 1$ when the transmission rate $\eta$ of the lossy GBS satisfies 
    \begin{align}\label{eq: loss bound for tvd error}
        (1-\eta)M\sinh^2 r \leq \beta^2 .
    \end{align}
\end{lemma}
\begin{proof}
    See Sec.~\ref{appendix: section: TVD between lossy GBS and ideal GBS} in the Supplemental Material.
\end{proof}

To sketch the proof, for any measurement basis, the total variation distance between two quantum states can be bounded by their quantum infidelity~\cite{fuchs1999cryptographic}.
Moreover, the quantum infidelity between the output states of ideal and lossy GBS can be formulated in terms of their covariance matrices~\cite{marian2012uhlmann, spedalieri2012limit}, with the latter depending on the transmission rate $\eta$.
Combining those two arguments yields Eq.~\eqref{eq: loss bound for tvd error}, which specifies the condition on $\eta$ to attain a desired total variation distance between ideal and lossy GBS.

In accordance with Theorem~\ref{theorem: total variation distance bound}, one can derive a condition for the classical hardness of simulating lossy GBS. 
Specifically, let $\eta$ satisfy Eq.~\eqref{eq: loss bound for tvd error}, such that the total variation distance between the ideal and lossy GBS is bounded by $\beta$. 
Then, if lossy GBS can be classically simulated within total variation distance $\beta_{1}$, the triangle inequality implies that the ideal GBS can also be simulated within distance $\beta + \beta_{1}$. 
Therefore, if classical simulation of the ideal GBS within total variation distance $\beta_{0}$ is hard, then simulating lossy GBS within distance $\beta_{1} \leq \beta_{0}$ is also hard, provided that $\beta \leq \beta_{0}-\beta_{1}$.

Combining Lemma~\ref{theorem: total variation distance bound} and the above argument results in the following theorem.

\begin{theorem}\label{corollary: simulation hardness}
    Suppose that $\beta_1 \leq \beta_{0} < 1$.
    Then, simulating lossy GBS within total variation distance $\beta_1$ is at least as hard as simulating ideal GBS within total variation distance $\beta_{0}$,  
    when the transmission rate $\eta$ of lossy GBS satisfies
    \begin{align}\label{eq: upper bound on number of lost photons}
        (1-\eta)M\sinh^2 r  \leq (\beta_0 - \beta_1)^2 .
    \end{align}
\end{theorem}



By definition, the left-hand side of Eq.~\eqref{eq: upper bound on number of lost photons} indicates the mean number of lost photons for lossy GBS with transmission rate $\eta$. 
Therefore, Theorem~\ref{corollary: simulation hardness} indicates that, when simulating ideal GBS within total variation distance $\beta_0$ is hard, simulating lossy GBS within the distance $\beta_{1}$ is also hard when the mean number of lost photons is smaller than $(\beta_0 - \beta_1)^2$.
Since the classical hardness of ideal GBS requires inverse polynomial total variation distance, the bound in Eq.~\eqref{eq: upper bound on number of lost photons} gives the hardness result of lossy GBS only when the mean number of lost photons is far less than one.
This indeed yields a looser bound on the tolerable number of lost photons compared to Theorem~\ref{thm: main result informal}, which establishes hardness even under logarithmically scaling photon loss.

Nevertheless, we emphasize that the bound in Eq.~\eqref{eq: upper bound on number of lost photons} is not optimal, as the inequalities employed, including the relation between total variation distance and quantum infidelity, and the triangle inequality for total variation distance, are not necessarily tight.
Hence, it still remains open that the bound in Eq.~\eqref{eq: upper bound on number of lost photons} could be further optimized, thus providing an alternative approach toward stronger classical hardness results for lossy GBS.





\emph{Conclusion.---}
In this work, we established complexity-theoretic foundations for the classical hardness of lossy GBS. 
We explicitly characterized the polynomial structure of lossy output probability in terms of transmission rate~$\eta$, and established an extrapolation scheme that allows estimation of the ideal output probability from lossy output probability estimates. 
Consequently, we showed that lossy GBS remains as hard as ideal GBS when the average number of lost photons scales at most logarithmically with the number of input photons.
This scaling of noisy photons for hardness improves upon the previous hardness result for lossy Fock state boson sampling with $O(1)$ number of lost photons~\cite{aaronson2016bosonsampling}, and reproduces the scaling obtained for partially distinguishable Fock state boson sampling~\cite{go2025quantum}.
Since GBS and photon loss respectively represent the central platform and dominant imperfection in current quantum advantage experiments~\cite{zhong2020quantum, zhong2021phase, madsen2022quantum, deng2023gaussian, liu2025robust, qi2020regimes, villalonga2021efficient, bulmer2022boundary, liu2023simulating, oh2024classical, umanskii2024classical}, we believe that our results provide an essential theoretical framework for achieving quantum advantage under realistic experimental conditions.

We remark on an existing hardness argument for lossy GBS and compare it with our main result. 
Specifically, using the threshold theorem, Ref.~\cite{deshpande2022quantum} shows the \#P-hardness of estimating output probabilities of lossy GBS within additive error $2^{-\text{poly}(N)}$ for a sufficiently large $\text{poly}(N)$.
In contrast, we establish the classical intractability of lossy GBS up to ideal GBS, and characterize the intractable regimes in which our hardness result holds.

We also note a key distinction between GBS and Fock state boson sampling under photon loss, underscoring the necessity of our hardness analyses. 
For lossy Fock state boson sampling, one can post-select no-loss events by keeping only outcomes with output photon number equal to the input photon number, which directly yields a hardness result when at most a logarithmic fraction of input photons are lost~\cite{aaronson2016bosonsampling, go2025quantum}.
In contrast, such a post-selection trick is not possible for lossy GBS; hence, the hardness of lossy GBS cannot be obtained in this straightforward manner, necessitating new approaches as presented in this work.

There still remain important open problems. 
First, stronger hardness results beyond the logarithmic-loss regime are required to ultimately ensure the classical intractability of simulating near-term GBS experiments~\cite{zhong2020quantum, zhong2021phase, madsen2022quantum, deng2023gaussian, liu2025robust}.
Also, establishing hardness results for partial distinguishability noise is crucial, which is also a dominant source of imperfections in current experimental setups~\cite{zhong2020quantum, zhong2021phase, madsen2022quantum, deng2023gaussian, liu2025robust, chen2023scalable, young2024atomic}.
This could further lead to improved hardness results for noisy GBS, accounting for both photon loss and partial distinguishability of photons.


\emph{Acknowledgement.---}
B.G. and H.J. were supported by the Korean government (Ministry of Science and ICT~(MSIT)),
the NRF grants funded by the Korea government~(MSIT)~(Nos.~RS-2024-00413957,~RS-2024-00438415, and NRF-2023R1A2C1006115),~and the Institute of Information \& Communications Technology Planning \& Evaluation (IITP) grant funded by the Korea government (MSIT) (IITP-2025-RS-2020-II201606 and  IITP-2025-RS-2024-00437191).
C.O. was supported by the NRF Grants (No. RS-2024-00431768 and No. RS-2025-00515456) funded by the Korean government (MSIT) and IITP grants funded by the Korea government (MSIT) (No. IITP-2025-RS-2025-02283189 and  IITP-2025-RS-2025-02263264).

\let\oldaddcontentsline\addcontentsline
\renewcommand{\addcontentsline}[3]{}
\bibliographystyle{unsrt}
\bibliography{Reference.bib}
\let\addcontentsline\oldaddcontentsline
\onecolumngrid

\clearpage
\widetext
\begin{center}
\textbf{\large Supplemental Material for ``Sufficient conditions for hardness of lossy Gaussian boson sampling''}
\end{center}

\setcounter{page}{1}
\renewcommand{\thesection}{S\arabic{section}}
\renewcommand{\theequation}{S\arabic{equation}}
\renewcommand{\thefigure}{S\arabic{figure}}
\renewcommand{\thetable}{S\arabic{table}}

\addtocontents{toc}{\protect\setcounter{tocdepth}{0}}

\tableofcontents


\clearpage

\section{Related results and our contribution}\label{section: related works}

This section remarks on the previous analyses related to our result and clarifies our contribution.

One of the most relevant results to our work is Ref.~\cite{aaronson2016bosonsampling}, which shows that lossy Fock state boson sampling with a fixed $O(1)$ number of lost photons over $N_0$ input photons maintains the same complexity level as the ideal Fock state boson sampling.
Subsequently, Ref.~\cite{go2025quantum} similarly establishes the hardness result for partial distinguishability noise, claiming that $O(\log N_0)$ distinguishable photons out of $N_0$ input photons still preserve the classical hardness of ideal Fock state boson sampling. 
Both analyses share a common high-level strategy: by exploiting the polynomial structure of noisy output probability, they construct an extrapolation scheme that enables the estimation of ideal output probability from noisy output probability estimates, thereby proving that the noisy sampler remains as hard to simulate as its ideal counterpart within a certain noise regime.

We extend this framework to GBS under photon loss by adopting a similar high-level strategy while overcoming several new technical challenges.
Specifically, we explicitly formulate the lossy output probability, characterize its polynomial dependence on the transmission rate $\eta$, and construct an extrapolation scheme that connects the lossy and ideal output probabilities.
We also establish a lower bound on the post-selection probability that allows us to fix the output photon number, which is crucial for the hardness proof of samplers with varying output photon numbers.
As a consequence, we prove that lossy GBS with an average photon loss of $O(\log N_0)$ out of $N_0$ input photons retains the same complexity-theoretic level as ideal GBS.
This scaling of noisy photons in our work improves upon the previous result for lossy Fock state boson sampling~\cite{aaronson2016bosonsampling}, and matches the scaling for partially distinguishable Fock state boson sampling~\cite{go2025quantum}.
As GBS serves as the primary experimental platform and photon loss remains the dominant imperfection in current quantum advantage experiments~\cite{zhong2020quantum, zhong2021phase, madsen2022quantum, deng2023gaussian, liu2025robust, qi2020regimes, villalonga2021efficient, bulmer2022boundary, liu2023simulating, oh2024classical, umanskii2024classical}, our result complements those previous hardness results~\cite{aaronson2016bosonsampling, go2025quantum} by addressing a more experimentally relevant setting for achieving quantum advantage.

Another related line of research is the classical simulability results under photon loss noise. 
In the asymptotic regime, Refs.~\cite{oszmaniec2018classical, garcia2019simulating} show the classical simulability results for the lossy Fock state boson sampling, claiming that it becomes efficiently simulable when $O(\sqrt{N_0})$ photons survive over $N_0$ input photons. 
Ref.~\cite{qi2020regimes} generalizes these results to lossy GBS, showing the classical simulability of lossy GBS when the average number of surviving photons scales $O(\sqrt{N_0})$ for $N_0$ input photons. 
In the finite-sized regime, by using the effect of photon loss on Gaussian boson sampling, Refs.~\cite{villalonga2021efficient, bulmer2022boundary, oh2024classical} present classical algorithms that can reproduce the recent Gaussian boson sampling experiments~\cite{zhong2020quantum, zhong2021phase, madsen2022quantum, deng2023gaussian}.
These simulability results suggest a noise threshold for the classical simulability of lossy GBS, providing a necessary condition for achieving quantum advantage. 
In contrast, our result provides a noise threshold for the classical \textit{intractability} of lossy GBS.
Definitely, this complements the existing simulability results by providing a sufficient condition for the loss rate to achieve quantum advantage.
Moreover, our work characterizes the complexity gap between the classically simulable noise regime with $O(\sqrt{N_0})$ surviving photons~\cite{qi2020regimes} and the classically intractable noise regime with $O(\log N_0)$ lost photons in lossy GBS.
This provides deeper insights into the computational complexity of lossy GBS and can help identify the precise transition point between the classical and quantum regimes.

\section{Deriving output probability distribution of lossy GBS}\label{Appendix:section:outputprobabilityoflossyGBS}

We first introduce the output distribution of an $M$-mode Gaussian state measured by the photon number basis, i.e., the general output distribution of GBS. 
Here, we only consider the mean vector of the Gaussian state as zero.
We denote $\Sigma$ as a covariance matrix in terms of $\alpha\alpha^*$ basis (as depicted in~\cite{hamilton2017gaussian}), and denote $\sigma$ as a covariance matrix in terms of $xp$ basis.
We denote $S = (s_1,\dots,s_N)$ as an $N$-photon outcome from GBS, such that $1\leq s_i\leq M$ and $s_1\le s_2\le\cdots\le s_N$, where $s_i$ corresponds to a mode number where an $i$th photon is measured.
For an $M$ mode Gaussian state with zero-mean and covariance matrix $\Sigma$, the output probability $p_S$ of obtaining an $N$-photon outcome $S$ can be expressed as~\cite{hamilton2017gaussian, kruse2019detailed} 
\begin{align}\label{appendix:outputprobabilityofgaussianstate}
    q_S = \frac{1}{\mu(S)}\frac{1}{\sqrt{|\Sigma_Q|}}\Haf(A_{S\oplus S}),
\end{align}
where $\mu(S)$ denotes a product of the multiplicity of all possible values in $S$ (i.e., $\mu(S) = 1$ for collision-free outcome $S$), $\Sigma_Q = \Sigma + \frac{1}{2}\mathbb{I}_{2M}$ denotes Q-representation of covariance matrix (where $\mathbb{I}_{2M}$ denotes $2M$ by $2M$ identity matrix), $|\Sigma_Q|$ = $\det(\Sigma_Q)$, and 
\begin{align}
    A = X_{2M}(\mathbb{I}_{2M} - \Sigma_Q^{-1}) \quad\text{where}\quad X_{2M} = \begin{pmatrix}
        0 & \mathbb{I}_{M} \\ 
        \mathbb{I}_{M} & 0
    \end{pmatrix},
\end{align}
and $A_{S\oplus S}$ is a $2N$ by $2N$ matrix defined by taking $j$-th and $(j + M)$-th rows and columns of $A$ for all $j \in S$.

\subsection{Without post-selection}

We first investigate the output probability of lossy GBS without post-selection.
To obtain the analytic form of the output probability, we aim to figure out $|\Sigma_Q|$ and $A$ in Eq.~\eqref{appendix:outputprobabilityofgaussianstate} for the output state of our lossy GBS setup. 
Specifically, we consider our lossy GBS setup as follows.
We first prepare input $M$ squeezed vacuum states ($p$-squeezing) with equal squeezing parameters $r$.
For this input state, the initial covariance matrix $\sigma_{0}$ in $xp$ basis is
\begin{align}\label{eq: input xp covariance matrix before loss}
    \sigma_{0} = \begin{pmatrix}
        \frac{ e^{2r}}{2} \mathbb{I}_{M} & 0 \\
        0 & \frac{ e^{-2r}}{2} \mathbb{I}_{M}
        \end{pmatrix} .
\end{align}
Then, this input state is evolved by an $M$-mode linear optical network (corresponding to $M$ by $M$ unitary matrix $U$). 
For the photon loss model, we consider the conventional beamsplitter loss model in Ref.~\cite{oszmaniec2018classical, garcia2019simulating, qi2020regimes} where the degree of noise is characterized by a uniform transmission rate $\eta$. 
Notice that, this loss channel can commute with any linear optical elements, and thus, we can instead consider all the loss occurring at the input state. 
Therefore, when photon loss with the uniform transmission rate $\eta$ is applied to the input state, the input covariance matrix $\sigma_{0}$ in $xp$ basis transforms to the noisy input covariance matrix $\sigma_{\text{in}}$ defined as~\cite{weedbrook2012gaussian}
\begin{align}\label{eq: input xp covariance matrix after loss}
\sigma_{0} \Rightarrow \sigma_{\text{in}} = \eta \sigma_{0} + (1-\eta)\frac{1}{2}\mathbb{I}_{2M} = \begin{pmatrix}
\frac{\eta e^{2r} + (1-\eta)}{2} \mathbb{I}_{M} & 0 \\
0 & \frac{\eta e^{-2r} + (1-\eta)}{2} \mathbb{I}_{M}
\end{pmatrix} ,
\end{align}
and accordingly, the corresponding Q-representation is
\begin{align}\label{appendix:sigmaQ}
\sigma_{\text{in},Q} = \sigma_{\text{in}} + \frac{1}{2}\mathbb{I}_{2M} = \begin{pmatrix}
\left(1 + \frac{\eta}{2}(e^{2r}-1)\right)\mathbb{I}_{M} & 0 \\
0 & \left(1 + \frac{\eta}{2}(e^{-2r}-1)\right)\mathbb{I}_{M}
\end{pmatrix} .
\end{align}
Note that the unitary transformation of the covariance matrix does not vary its determinant. 
Therefore, the determinant of the output covariance matrix $|\Sigma_Q|$ after unitary evolution by $U$ is equivalent to $|\sigma_{\text{in},Q}|$, and using Eq.~\eqref{appendix:sigmaQ}, one can find that
\begin{align}
    |\Sigma_Q| &= |\sigma_{\text{in},Q}| \\
    &= (1 + \frac{\eta}{2}(e^{2r}-1))^{M}(1 + \frac{\eta}{2}(e^{-2r}-1))^{M} \\
    &= \left((1-\frac{\eta}{2})e^{-r} + \frac{\eta}{2}e^r\right)^M\left((1-\frac{\eta}{2})e^{r} + \frac{\eta}{2}e^{-r}\right)^M \\
    &= \left(\frac{e^r + e^{-r}}{2} - (1 - \eta)\frac{e^r - e^{-r}}{2}\right)^M\left(\frac{e^r + e^{-r}}{2} + (1 - \eta)\frac{e^r - e^{-r}}{2}\right)^M \\
    &= \cosh^{2M}{r}\left(1 - (1-\eta)^2\tanh^2{r} \right)^M .
\end{align}
We now examine $A$ matrix for the output state of our lossy GBS setup.
From Eq.~\eqref{appendix:sigmaQ}, we have
\begin{align}
\mathbb{I}_{2M} - \sigma_{\text{in},Q}^{-1} = \begin{pmatrix}
\frac{e^{2r}-1}{e^{2r} + \frac{2}{\eta}-1}\mathbb{I}_M & 0 \\
0 & \frac{e^{-2r}-1}{e^{-2r} + \frac{2}{\eta} - 1}\mathbb{I}_M
\end{pmatrix} .
\end{align}
By changing the basis from $xp$ to $\alpha\alpha^*$, the above term transforms as
\begin{align}
\mathbb{I}_{2M} - \Sigma_{\text{in},Q}^{-1} &= \frac{1}{(e^{r} + (\frac{2}{\eta}-1)e^{-r}) (e^{-r} + (\frac{2}{\eta}-1)e^{r})}\begin{pmatrix}
(\frac{1}{\eta}-1)(e^r-e^{-r})^2\mathbb{I}_M & \frac{1}{\eta}(e^{2r}-e^{-2r})\mathbb{I}_M \\
\frac{1}{\eta}(e^{2r}-e^{-2r})\mathbb{I}_M  & (\frac{1}{\eta}-1)(e^r-e^{-r})^2\mathbb{I}_M
\end{pmatrix} \\
&=  \frac{e^{2r}-e^{-2r}}{\eta(e^{r} + (\frac{2}{\eta}-1)e^{-r}) (e^{-r} + (\frac{2}{\eta}-1)e^{r})}\begin{pmatrix}
(1-\eta)\tanh{r}\mathbb{I}_M & \mathbb{I}_M \\
\mathbb{I}_M  & (1-\eta)\tanh{r}\mathbb{I}_M
\end{pmatrix} \\
&=  \frac{\eta(e^{r}-e^{-r})(e^{r}+e^{-r})}{(e^{r}+e^{-r} -(1-\eta)(e^{r}-e^{-r})) (e^{r}+e^{-r} + (1-\eta)(e^{r}-e^{-r}))}\begin{pmatrix}
(1-\eta)\tanh{r}\mathbb{I}_M & \mathbb{I}_M \\
\mathbb{I}_M  & (1-\eta)\tanh{r}\mathbb{I}_M
\end{pmatrix} \\
&= \frac{\eta\tanh{r}}{1-(1-\eta)^2\tanh^2{r}}\begin{pmatrix}
(1-\eta)\tanh{r}\mathbb{I}_M & \mathbb{I}_M \\
\mathbb{I}_M  & (1-\eta)\tanh{r}\mathbb{I}_M
\end{pmatrix} .
\end{align}
Finally, by applying unitary transformation in $\alpha\alpha^*$ basis that corresponds to the unitary matrix $U$, we obtain
\begin{align}
\mathbb{I}_{2M} - \Sigma_{Q}^{-1} &= \begin{pmatrix}
U &  0\\
0   & U^{*}
\end{pmatrix}
\left(\mathbb{I}_{2M} - \Sigma_{\text{in},Q}^{-1}\right)\begin{pmatrix}
U^{\dag} &  0\\
0   & U^{T}
\end{pmatrix} \\
&= \frac{\eta\tanh{r}}{1-(1-\eta)^2\tanh^2{r}}\begin{pmatrix}
(1-\eta)\tanh{r}\mathbb{I}_M & UU^T  \\
U^*U^{\dag} & (1-\eta)\tanh{r}\mathbb{I}_M
\end{pmatrix} ,
\end{align}
and thus the matrix $A$ can be represented as
\begin{align}
A = \frac{\eta\tanh{r}}{1-(1-\eta)^2\tanh^2{r}}\begin{pmatrix}
U^*U^{\dag} & (1-\eta)\tanh{r}\mathbb{I}_M \\
(1-\eta)\tanh{r}\mathbb{I}_M  & UU^T
\end{pmatrix} .
\end{align}
Substituting these results in Eq.~\eqref{appendix:outputprobabilityofgaussianstate}, the output probability of lossy GBS to obtain $N$-photon outcome $S$ can be expressed as
\begin{equation}\label{appendix: eq: output probability without postselection}
    q_S(\eta, U) = \frac{1}{\mu(S)}\frac{\eta^{N}\tanh^{N}{r}}{\cosh^{M}{r}(1-(1-\eta)^2\tanh^2{r})^{\frac{M}{2}+N}} \Haf{\begin{pmatrix}
(UU^{T})_{S} & (1-\eta)\tanh{r}\mathbb{I}_{N} \\
(1-\eta)\tanh{r}\mathbb{I}_{N} & (U^{*}U^{\dag})_{S} 
\end{pmatrix}} ,
\end{equation}
where $(UU^{T})_{S}$ is an $N$ by $N$ matrix defined by taking rows and columns of $UU^{T}$ according to $S$ (similarly for $(U^*U^{\dag})_S$).
Here, we interchanged the location of $(UU^{T})_{S}$ and $(U^*U^{\dag})_S$ for convenience, employing the fact that the output hafnian value is invariant under such transformation.

\subsection{With post-selection}

We now derive the output probability of lossy GBS with post-selecting $N$-photon outcomes, which we use throughout this work for the hardness arguments.
To do so, we first figure out the probability $\Pr[N]$ of generating $N$ photons from our lossy GBS setup, which can be obtained by considering all possible input photon configurations.
Note that throughout our analysis, we only consider $N$ as even, but one can easily extend this result to the odd $N$ case.

Considering that only an even number of input photons is allowed for squeezed vacuum inputs, to obtain a total $N$ output photons in the lossy environment, the possible input photon numbers over $M$ equally squeezed vacuum states are $N$, $N+2$, $N+4$, and so on. 
For $M$ equally squeezed vacuum states, the probability to generate $N + 2n$ input photons (where $n \in \{0,1,2,\dots\}$) follows negative binomial distribution as~\cite{hilbe2011negative}
\begin{align}
\frac{\tanh^{N+2n}{r}}{\cosh^{M}{r}}\binom{\frac{M}{2} + \frac{N}{2} + n - 1}{\frac{N}{2} + n}.
\end{align}
Because the final output photon number is given as $N$, among the $N+2n$ input photons, $N$ photons survive with each probability $\eta$ and $2n$ photons are lost with each probability $1 - \eta$.
Hence, for $N + 2n$ input photons, the probability to obtain $N$ output photons is given by the binomial distribution as $\binom{N + 2n}{2n}\eta^N(1-\eta)^{2n}$. 
Combining these arguments, $\Pr[N]$ can be represented as
\begin{align}
   \Pr[N] &= \eta^{N}\frac{\tanh^{N}r}{\cosh^{M}r}\sum_{n=0}^{\infty} \binom{\frac{M}{2} + \frac{N}{2} + n -1 }{\frac{N}{2} + n}\binom{N+2n}{2n}(1-\eta)^{2n}\tanh^{2n}r  \\ 
   &=  \eta^{N}\frac{\tanh^{N}r}{\cosh^{M}r}\sum_{n=0}^{\infty} \frac{\left( \frac{M}{2} + \frac{N}{2} + n -1 \right)! }{\left( \frac{N}{2} + n \right)! \left( \frac{M}{2} - 1\right)!} \frac{(N + 2n)!}{(2n)! N!} (1-\eta)^{2n}\tanh^{2n}r   \\
   &=  \eta^{N}\frac{\tanh^{N}r}{\cosh^{M}r} \frac{\left( \frac{M}{2} + \frac{N}{2} - 1 \right)!}{\left( \frac{N}{2} \right)! \left( \frac{M}{2} - 1\right)!}\sum_{n=0}^{\infty} \frac{\left( \frac{M}{2} + \frac{N}{2} + n -1 \right)! }{\left( \frac{M}{2} + \frac{N}{2} - 1 \right)! } \frac{(N + 2n - 1)!!}{(2n-1)!! (N-1)!!} \frac{(1-\eta)^{2n}\tanh^{2n}r}{n!}  \\ 
   &=  \eta^{N}\frac{\tanh^{N}r}{\cosh^{M}r} \binom{\frac{M}{2} + \frac{N}{2} -1 }{\frac{N}{2}} \sum_{n=0}^{\infty} \frac{ \left( \frac{M + N}{2} \right)_{n} \left(  \frac{N + 1}{2} \right)_n }{ \left( \frac{1}{2} \right)_n } \frac{(1-\eta)^{2n}\tanh^{2n}r}{n!}  , \label{eq: asdeee}
\end{align}
where $(q)_n$ in the right-hand side of Eq.~\eqref{eq: asdeee} is the (rising) Pochhammer symbol defined as 
\begin{align}\label{eq: pochhammer}
    (q)_n = \begin{cases}
1 & \text{when}\;\; n = 0 \\
q(q+1)\cdots(q+n-1) &  \text{when}\;\; n \in \mathbb{Z}^{+} .
\end{cases} 
\end{align}
In fact, the summation term in the right-hand side of Eq.~\eqref{eq: asdeee} is identical to the hypergeometric function ${}_2F_1(a,b;c;z)$~\cite{whittaker1920course}.
Accordingly, we obtain the final analytic expression for $\Pr[N]$ as
\begin{align}
    \Pr[N] = \eta^{N}\frac{\tanh^{N}r}{\cosh^{M}r}\binom{\frac{M}{2} + \frac{N}{2} -1 }{\frac{N}{2}} {}_2F_1\left(\frac{M + N}{2}, \frac{N + 1}{2}; \frac{1}{2};(1-\eta)^2\tanh^2r\right). \label{appendix:postselectionprobability}
\end{align}
We remark that the above expression agrees with the post-selection probability of obtaining even $N$-photon outcomes of lossy GBS described in Ref.~\cite[Supplementary Materials]{deshpande2022quantum}.

Therefore, combining Eq.~\eqref{appendix: eq: output probability without postselection} and Eq.~\eqref{appendix:postselectionprobability}, the output probability for $N$-photon outcome $S$ after post-selecting $N$-photon outcomes can finally be expressed as 
\begin{align}\label{supple: noisy output probability}
p_{S}(\eta,U) = \frac{1}{\mu(S)}\binom{\frac{M}{2} + \frac{N}{2} -1 }{\frac{N}{2}}^{-1} Q(\eta)^{-1} \Haf{\begin{pmatrix}
(UU^{T})_{S} & (1-\eta)\tanh{r}\mathbb{I}_{N} \\
(1-\eta)\tanh{r}\mathbb{I}_{N} & (U^{*}U^{\dag})_{S} 
\end{pmatrix}}    ,
\end{align}
where the multiplicative factor $Q(\eta)$ is defined as 
\begin{align}
Q(\eta) = (1-(1-\eta)^2\tanh^2{r})^{\frac{M}{2}+N}{}_2F_1\left(\frac{M + N}{2}, \frac{N + 1}{2}; \frac{1}{2};(1-\eta)^2\tanh^2r\right),
\end{align}
thus obtaining the expression in the main text.
Note that in the limit $\eta = 1$, which corresponds to ideal GBS, $Q(1) = 1$, and Eq.~\eqref{supple: noisy output probability} recovers the (post-selected) ideal output probability~\cite{hamilton2017gaussian, kruse2019detailed, deshpande2022quantum} such that
\begin{align}\label{supple: ideal output probability}
    p_{S}(U) \coloneqq p_{S}(1,U) =  \frac{1}{\mu(S)}\binom{\frac{M}{2} + \frac{N}{2} -1 }{\frac{N}{2}}^{-1} \Haf{\begin{pmatrix}
(UU^{T})_{S} & 0 \\
0 & (U^{*}U^{\dag})_{S} 
\end{pmatrix}} .
\end{align}

\section{Post-selection probability for mean output photon number}\label{Appendix:section:postselectionprobabilityoflossyGBS}

In this section, we derive the lower bound of the post-selection probability $\Pr[N]$ given in Eq.~\eqref{appendix:postselectionprobability} when $N$ is the mean output photon number $N = \eta M\sinh^{2}{r}$ in our lossy GBS setup.

First, observe that 
\begin{align}
\eta^{N}\frac{\tanh^{N}r}{\cosh^{M}r}\binom{\frac{M}{2} + \frac{N}{2} -1 }{\frac{N}{2}}   &=  \frac{\eta^{\frac{N}{2}}(\frac{N}{M})^{\frac{N}{2}}}{(1+\frac{N}{\eta M})^{\frac{M+N}{2}}}\binom{\frac{M}{2} + \frac{N}{2} -1 }{\frac{N}{2}} \label{tyty} \\
&\ge \sqrt{\frac{\frac{M}{2} + \frac{N}{2} - 1}{\pi N \left( \frac{M}{2} - 1\right)}}e^{\frac{1}{6(M+N-2) + 1} - \frac{1}{6N} - \frac{1}{6(M-2)}} \frac{\eta^{\frac{N}{2}}(\frac{N}{M})^{\frac{N}{2}} (\frac{M+N-2}{N})^{\frac{N}{2}}(\frac{M+N-2}{M-2})^{\frac{M}{2}-1}}{(1+\frac{N}{\eta M})^{\frac{M+N}{2}}} \label{qwer}\\
&\geq \Omega\left(\frac{1}{\sqrt{N}}e^{-\frac{N}{2}(\frac{1}{\eta}-1-\log\eta)}\right) ,\label{qwerty}
\end{align}
where we used $\cosh^2{r} = \frac{N}{\eta M} + 1$ and $\tanh^2{r} = \frac{N}{\eta M + N}$ in Eq.~\eqref{tyty},
used two-sided Stirling's inequality~\cite{robbins1955remark} in Eq.~\eqref{qwer}, and used the following relation in Eq.~\eqref{qwerty}:
\begin{align}
    &\frac{\eta^{\frac{N}{2}}(\frac{N}{M})^{\frac{N}{2}} (\frac{M+N-2}{N})^{\frac{N}{2}}(\frac{M+N-2}{M-2})^{\frac{M}{2}-1}}{(1+\frac{N}{\eta M})^{\frac{M+N}{2}}} \nonumber \\
    &= \exp[\frac{N}{2}\log\eta + \frac{N}{2}\log( 1 + \frac{N-2}{M}) +  \frac{M-2}{2}\log(1 + \frac{N}{M-2}) - \frac{M+N}{2}\log(1 + \frac{N}{\eta M})   ] \\
    &\geq \exp[ \frac{N}{2}\log \eta + \frac{N}{2}\left( \frac{N-2}{M} - \frac{1}{2}\left( \frac{N-2}{M} \right)^2 \right) + \frac{M-2}{2}\left( \frac{N}{M-2} - \frac{1}{2}\left( \frac{N}{M-2} \right)^2 \right) - \frac{M+N}{2} \left( \frac{N}{\eta M} \right)] \\
    & = \exp[ -\frac{N}{2}\left( \frac{1}{\eta} - 1 - \log\eta \right) + o(1)],
\end{align}
where we used the fact that $\frac{N^2}{M} = o(1)$ (following from our setting $M \propto N^{\gamma}$ with $\gamma > 2$) and $x - \frac{x^2}{2} \leq \log (1+x) \leq x$ for $x \geq 0$.

Meanwhile, the hypergeometric function in Eq.~\eqref{appendix:postselectionprobability} can be expressed as 
\begin{align}\label{appendix:hypergeometric}
{}_2F_1\left(\frac{M + N}{2}, \frac{N + 1}{2}; \frac{1}{2};(1-\eta)^2\tanh^2r\right) = \sum_{n=0}^{\infty}\frac{(\frac{M + N}{2})_n (\frac{N + 1}{2})_n}{(\frac{1}{2})_n} \frac{(1-\eta)^{2n}\tanh^{2n}{r}}{n!},
\end{align}
for $(q)_n$ being the Pochhammer symbol defined in Eq.~\eqref{eq: pochhammer}. 
We find that each term in the summation in Eq.~\eqref{appendix:hypergeometric} can be lower bounded as
\begin{align}
\frac{(\frac{M + N}{2})_n (\frac{N + 1}{2})_n}{(\frac{1}{2})_n}  \frac{(1-\eta)^{2n}\tanh^{2n}{r}}{n!} &= \frac{(M+N+2n-2)!!(N+2n-1)!!}{(2n)!(M+N-2)!!(N-1)!!}(1-\eta)^{2n}\tanh^{2n}{r} \\
&\ge \frac{M^nN^n}{(2n)!}(1-\eta)^{2n}\tanh^{2n}{r},
\end{align}
which leads to 
\begin{align}
    {}_2F_1\left(\frac{M + N}{2}, \frac{N + 1}{2}; \frac{1}{2};(1-\eta)^2\tanh^2r\right) &\ge \sum_{n=0}^{\infty}\frac{M^nN^n}{(2n)!}(1-\eta)^{2n}\tanh^{2n}{r} \\
    &= \cosh\left({\sqrt{MN}(1-\eta)\tanh{r}}\right) \\
    &= \Omega(e^{\sqrt{MN}(1-\eta)\tanh{r}}) . \label{aszxcx}
\end{align}
Combining those arguments, the post-selection probability $\Pr[N]$ for mean photon number $N = \eta M\sinh^{2}{r}$ can be bounded from below as 
\begin{align}\label{appendix:lowerboundofpostselectionprobabilitygeneral}
    \Pr[N = \eta M\sinh^{2}{r}] \ge \Omega\left(\frac{1}{\sqrt{N}}e^{-\frac{N}{2}(\frac{1}{\eta}-1-\log\eta) + \sqrt{MN}(1-\eta)\tanh{r}}\right).
\end{align}

To proceed further, we note that we are considering the \textit{low-loss} regime for our main hardness result of lossy GBS (e.g., as stated in Theorem~\ref{maintheorem}), such that $\eta$ satisfies $(1-\eta) \leq O(\frac{1}{\sqrt{N}})$.
Then, observe that
\begin{align}\label{qqwer}
\frac{N}{2}\left(\frac{1}{\eta}-1-\log\eta \right) = \frac{N}{2}\left(\frac{1-\eta}{1-(1-\eta)}-\log(1-(1-\eta)) \right) =   (1-\eta)N + O(1) .  
\end{align} 
Also, using the fact that $\tanh^2{r} = \frac{N}{\eta M + N}$ and $N = o(\sqrt{M})$, given that $(1-\eta) \leq O(\frac{1}{\sqrt{N}})$, we get
\begin{align}
\sqrt{MN}(1-\eta)\tanh{r} = (1 - \eta)N\frac{1}{\sqrt{1 - (1 - \eta) + \frac{N}{M}}} = (1 - \eta)N + O(1) .
\end{align}
Hence, given that $(1-\eta) \leq O(\frac{1}{\sqrt{N}})$, the post-selection probability $\Pr[N]$ to obtain the mean photon number $N = \eta M\sinh^{2}{r}$ can finally be bounded by
\begin{align}\label{appendix:lowerboundofpostselectionprobability}
    \Pr[N = \eta M\sinh^{2}{r}] \ge \Omega\left(\frac{1}{\sqrt{N}}\right) ,
\end{align}
thus obtaining the desired bound.

\section{Hardness analysis of lossy GBS}\label{supplement: main result}

In this section, we give a formalized proof of Theorem~\ref{thm: main result informal} in the main text. 
To this end, we begin by formally defining the problems of average-case output probability estimation in both the ideal and lossy GBS cases.
Based on these definitions, we present a more formal statement of Theorem~\ref{thm: main result informal}, restated as Theorem~\ref{maintheoreminformal}.
Lastly, we provide a step-by-step proof of Theorem~\ref{maintheoreminformal}.

\subsection{Formulation of problems}

In the following, we formalize our main problems and restate Theorem~\ref{thm: main result informal}.
Before proceeding, let us first review the hiding property of GBS.
As depicted in~\cite{jiang2009entries, shou2025proof}, for a sufficiently large $M \gg N$, the distribution of $N$ by $N$ sub-matrices of symmetric product $UU^T$ for $M$ by $M$ Haar-random unitary matrix $U$ has bounded total variation distance (and converges asymptotically) to the distribution of $\frac{1}{M}XX^T$ for an $N$ by $M$ random Gaussian matrix $X \sim \mathcal{N}(0,1)_{\mathbb{C}}^{N\times M}$.
Since we are considering the dilute regime $M \propto N^{\gamma}$ with a sufficiently large $\gamma > 2$ in this work, the ideal output probability $p_{S}(U)$ in Eq.~\eqref{supple: ideal output probability} for a collision-free outcome $S$ and a Haar-random unitary circuit $U$ can also be expressed as
\begin{align}
    p_{S}(U) = \binom{\frac{M}{2} + \frac{N}{2} -1 }{\frac{N}{2}}^{-1} \Haf{\begin{pmatrix}
(UU^{T})_{S} & 0 \\
0 & (U^{*}U^{\dag})_{S} 
\end{pmatrix}} \simeq \binom{\frac{M}{2} + \frac{N}{2} -1 }{\frac{N}{2}}^{-1}M^{-N}\Haf{\begin{pmatrix}
XX^{T} & 0 \\
0 & X^{*}X^{\dag} 
\end{pmatrix}} 
\end{align}
for a random Gaussian matrix $X \sim \mathcal{N}(0,1)_{\mathbb{C}}^{N\times M}$. 
To simplify the analysis, we denote
\begin{align}\label{eq: rescaled ideal output probability}
    P(X) \coloneqq  \Haf{\begin{pmatrix}
XX^{T} & 0 \\
0 & X^{*}X^{\dag} 
\end{pmatrix}}
\end{align}
as a rescaled output probability of ideal GBS, which we will use for our hardness argument.

By this property, one can estimate $p_{S}(U)$ for most of Haar-random unitary matrix $U$, if one can estimate $P(X)$ for most of $X \sim \mathcal{N}(0,1)_{\mathbb{C}}^{N\times M}$, with rescaling factor $M^{-N}$ on the additive imprecision. 
Then, as we describe in the main text, by showing that estimating most of $P(X)$ over $X \sim \mathcal{N}(0,1)_{\mathbb{C}}^{N\times M}$ to within a certain additive imprecision is \#P-hard, the classical hardness of the ideal GBS can be obtained (see Sec.~\ref{supplement: hardness proof of GBS} or Refs.~\cite{aaronson2011computational, hamilton2017gaussian, kruse2019detailed} for details). 
We formally restate this average-case estimation problem of $P(X)$, which is also referred to as the problem of Gaussian Hafnian Estimation~\cite{hamilton2017gaussian, kruse2019detailed}. 

\begin{problem}[$|\text{GHE}|_{\pm}^{2}$~\cite{kruse2019detailed}]\label{def: GHE}
Given as input a matrix $X \sim \mathcal{N}(0,1)_{\mathbb{C}}^{N\times M}$ of i.i.d. Gaussians, together with error bounds $\epsilon_0,\,\delta_0 > 0$, estimate $P(X)$ to within additive error $\pm \epsilon_0\binom{\frac{M}{2} + \frac{N}{2} - 1}{\frac{N}{2}}N!$ with probability at least $1-\delta_0$ over $X$ in $\text{poly}(N,\epsilon_0^{-1},\delta_0^{-1})$ time. 
\end{problem}

We remark that the size of the additive error $\epsilon_0 \binom{\frac{M}{2} + \frac{N}{2} - 1}{\frac{N}{2}}N!$ in Problem~\ref{def: GHE} is different from that of~\cite{kruse2019detailed}.
This comes from the difference of their setup to ours, such that we are considering $M$ squeezed vacuum inputs over $M$ modes (i.e., full-mode input), while Ref.~\cite{kruse2019detailed} considers $N$ squeezed vacuum inputs over $M \propto  N^2$ modes (i.e., sparse input).



Importantly, provided that $|\rm{GHE}|_{\pm}^{2}$ is \#$\rm{P}$-hard, one can show that the ideal GBS cannot be classically simulated within total variation distance error in polynomial time, unless the infinite level of polynomial hierarchy collapses to a finite level; see Sec.~\ref{supplement: hardness proof of GBS} for details. 
In short, \#P-hardness of $|\rm{GHE}|_{\pm}^{2}$ is sufficient to establish the classical intractability of ideal GBS~\cite{hamilton2017gaussian, kruse2019detailed}.


Based on this understanding, one can simply extend this hardness analysis for the lossy GBS by replacing $P(X)$ with its noisy version in the $|\text{GHE}|_{\pm}^{2}$ problem.
More specifically, similarly to the rescaled ideal output probability in Eq.~\eqref{eq: rescaled ideal output probability}, we newly define $P(\eta,X)$ as a rescaled version of the noisy output probability $p_S(\eta,U)$ in Eq.~\eqref{supple: noisy output probability} multiplied by the same multiplicative factor as in Eq.~\eqref{eq: rescaled ideal output probability}, such that
\begin{align}\label{def: rescaled noisy output probability}
\begin{split}
P(\eta,X) \coloneqq Q(\eta)^{-1} \Haf{\begin{pmatrix}
XX^T & (1-\eta)M\tanh{r}\mathbb{I}_{N} \\
(1-\eta)M\tanh{r}\mathbb{I}_{N} & X^{*}X^{\dag}
\end{pmatrix}}    ,
\end{split}
\end{align}
where now $P(\eta,X)$ is a noisy analogue of $P(X)$ in Eq.~\eqref{eq: rescaled ideal output probability} such that $P(\eta,X)$ reduces to $P(X)$ at $\eta=1$.

Using this rescaled noisy output probability, we can generalize $|\text{GHE}|_{\pm}^{2}$ for the lossy GBS case, i.e., 
average-case estimation of $P(\eta, X)$ over $X \sim \mathcal{N}(0,1)_{\mathbb{C}}^{N\times M}$.
We now formally state our main problem of estimating $P(\eta, X)$, which we refer to as Gaussian Hafnian Estimation with Diagonals, motivated by the presence of additional diagonal elements in Eq.~\eqref{def: rescaled noisy output probability}. 

\begin{problem}[${\rm{GHED}}_{\pm}$]\label{def: sumGHE}
Given as input a transmission rate $\eta \in [0, \eta^*]$ for a fixed $\eta^*$ and a matrix $X \sim \mathcal{N}(0,1)_{\mathbb{C}}^{N\times M}$ of i.i.d. Gaussians, together with error bounds $\epsilon,\,\delta > 0$, estimate $P(\eta,X)$ to within additive error $\pm \epsilon\binom{\frac{M}{2} + \frac{N}{2} - 1}{\frac{N}{2}}N!$ with probability at least $1-\delta$ over $X$ in $\text{poly}(N,\epsilon^{-1},\delta^{-1})$ time. 
\end{problem}

Here, the difference from $|\text{GHE}|_{\pm}^{2}$ is that now there is an additional input noise parameter, a transmission rate $\eta$ below a fixed $\eta^*$. 
Then, analogously to the ideal GBS case, \#P-hardness of ${\rm{GHED}}_{\pm}$ leads to the classical hardness of lossy GBS that takes a transmission rate $\eta \in [0, \eta^*]$ as an additional input.
Accordingly, establishing the \#P-hardness of ${\rm{GHED}}_{\pm}$ suffices to prove the classical hardness of lossy GBS that takes a transmission rate $\eta \in [0, \eta^*]$ as an input, and finally, under Conjecture~\ref{conjecture: complexity decrease} in the main text, the classical hardness of lossy GBS with a fixed transmission rate $\eta^*$.

To derive \#P-hardness of ${\rm{GHED}}_{\pm}$ problem, we adopt a similar strategy in the previous hardness arguments~\cite{aaronson2016bosonsampling, go2025quantum}.
More precisely, we identify the minimum value $\eta_{\rm{th}}$ of $\eta^*$ such that ${\rm{GHED}}_{\pm}$ is \textit{polynomially reducible} from $|\text{GHE}|_{\pm}^{2}$.
which implies that solving ${\rm{GHED}}_{\pm}$ is as hard as solving $|\text{GHE}|_{\pm}^{2}$ (note that this is trivial for $\eta^* = 1$).
Then, summarizing all the previous arguments, lossy GBS with its transmission rate $\eta^*$ above the threshold $\eta_{\rm{th}}$ is equivalently hard to classically simulate as the ideal GBS case, since the classical hardness of the ideal GBS is grounded in the \#P-hardness of $|\text{GHE}|_{\pm}^{2}$ (i.e., Conjecture~\ref{conjecture: sharp-p-hard} in the main text).

To further simplify the analysis, we newly introduce a parameter denoted as $k$ and parameterize the transmission rate $\eta$ with $k$ such that
\begin{align}\label{eq: convention for eta}
\eta  = 1 - \frac{k}{N+k},
\end{align}
which equivalently gives $k = (\frac{1}{\eta}-1)N$.
To obtain physical intuition on $k$, note that when $N$ corresponds to the mean number of output photons $\eta M\sinh^2 r$ for general lossy GBS (without post-selection), then $k$ corresponds to the \textit{mean number of lost photons} $(1-\eta)M\sinh^2 r$.

Observe that given $\eta$, $k$ is uniquely determined via Eq.~\eqref{eq: convention for eta} (and vice versa).
Hence, given that the input $\eta$ is non-zero, $k \in [k^*, \infty)$ can be regarded as an alternative input variable to ${\rm{GHED}}_{\pm}$ in place of $\eta \in [0, \eta^*]$, where the minimum $k^*$ is determined by the maximum $\eta^*$ such that
\begin{align}\label{appendix: convention for eta-star}
\eta^* \coloneqq 1 - \frac{k^*}{N+k^*} .
\end{align}


Based on the arguments so far, we now restate Theorem~\ref{thm: main result informal} given in the main text.
We find that there exists a threshold $\eta_{\rm{th}}$ (or $k_{\rm{th}}$) such that solving ${\rm{GHED}}_{\pm}$ is complexity-theoretically equivalent to solving $|\rm{GHE}|_{\pm}^{2}$, when $\eta^* \geq \eta_{\rm{\rm{th}}}$ for $(1-\eta_{\rm{\rm{th}}})N = O(\log N)$, or equivalently, $k^* \leq k_{\rm{th}} = (\frac{1}{\eta_{\rm{th}}} - 1)N = O(\log N)$.

\begin{theorem}[Restatement of Theorem~\ref{thm: main result informal}]\label{maintheoreminformal}
There exists a threshold transmission rate $\eta_{\rm{th}}$ satisfying $(1-\eta_{\rm{\rm{th}}})N = O(\log N)$ such that for $\eta^* \geq \eta_{\rm{\rm{th}}}$, ${\rm{GHED}}_{\pm}$ is polynomially reducible from $|\rm{GHE}|_{\pm}^{2}$. 
In other words, if $\mathcal{O}$ is an oracle that solves ${\rm{GHED}}_{\pm}$ for $k^* \leq k_{\rm{th}} = O(\log N)$, together with error bounds given by $\epsilon, \delta = O({\rm{poly}}(N,\epsilon_0^{-1},\delta_0^{-1})^{-1})$, then $|\rm{GHE}|_{\pm}^{2}$ can be solved in $\rm{BPP}^{\mathcal{O}}$.
\end{theorem}

We provide a detailed proof of Theorem~\ref{maintheoreminformal} in the following section.
Before proceeding, we emphasize that for lossy GBS with mean output photon number $N = \eta^* M \sinh^2 r$, the condition $k^* \leq k_{\rm{th}} = O(\log N)$ in Theorem~\ref{maintheoreminformal} implies that the mean number of lost photons is at most logarithmically related to the mean number of input photons.
Therefore, Theorem~\ref{maintheoreminformal} implies that lossy GBS with at most a logarithmic number of lost photons is as hard as ideal GBS. 
This scaling coincides with that obtained for noisy Fock state boson sampling under partial distinguishability of photons~\cite{go2025quantum}.





\subsection{Establishing complexity-theoretic reduction}\label{Section: reduction}

In this section, we prove Theorem~\ref{maintheoreminformal} by explicitly constructing a complexity-theoretic reduction from $|\rm{GHE}|_{\pm}^{2}$ (Problem~\ref{def: GHE}) to ${\rm{GHED}}_{\pm}$ (Problem~\ref{def: sumGHE}).
We first outline the reduction procedure from $|\rm{GHE}|_{\pm}^{2}$ to ${\rm{GHED}}_{\pm}$, and then present a detailed proof of Theorem~\ref{maintheoreminformal}.

\subsubsection{Proof sketch}

In the following, we outline the reduction process from $|\text{GHE}|_{\pm}^{2}$ to ${\rm{GHED}}_{\pm}$, which follows the similar reduction process in~\cite{go2025quantum}. 
More precisely, we sketch how one can estimate $P(\eta=1,X) = P(X)$ value with high probability, given access to estimated values of $P(\eta,X)$ for $\eta \in [0, \eta^*]$, i.e., given the oracle access to the ${\rm{GHED}}_{\pm}$ problem.

To proceed, we rescale again the noisy output probability $P(\eta, X)$ in Eq.~\eqref{def: rescaled noisy output probability} without the multiplicative factor $Q(\eta)^{-1}$ at the front, which we will denote as $R_{X}(\eta)$.
In other words, we have
\begin{align}
    R_{X}(\eta) &\coloneqq Q(\eta)P(\eta, X) 
     = \Haf{\begin{pmatrix}
XX^T & (1-\eta)M\tanh{r}\mathbb{I}_{N} \\
(1-\eta)M\tanh{r}\mathbb{I}_{N} & X^{*}X^{\dag}
\end{pmatrix}}    .
\end{align}
Notice that now $R_{X}(\eta)$ is an $N$-degree polynomial in $\eta$, and reduces to the ideal output probability at $\eta = 1$ such that $R_{X}(1) = P(X)$.
Moreover, we find in Sec.~\ref{appendix:section:efficient} that the multiplicative factor $Q(\eta)$ in the above is efficiently computable, up to a desired accuracy.
Hence, given an estimation value of $P(\eta,X)$ for any $\eta \in [0, \eta^*]$ (i.e., given an oracle access to ${\rm{GHED}}_{\pm}$), we can also obtain an estimation value of $R_{X}(\eta)$ for any $\eta \in [0, \eta^*]$ by multiplying $Q(\eta)$.

Accordingly, our basic idea is to infer $R_{X}(1)$ value via \textit{polynomial interpolation}, using the estimation values of $R_{X}(\eta)$ for $\eta \in [0, \eta^*]$ obtained by the oracle access to ${\rm{GHED}}_{\pm}$.
That is, we first estimate $R_{X}(\eta)$ values for different values of $\eta \le \eta^*$, draw a polynomial using these values, and finally estimate $R_{X}(\eta)$ at $\eta = 1$.
However, since polynomial interpolation induces exponentially increasing imprecision level with the degree of the polynomial~\cite{paturi1992degree, kondo2022quantum, bouland2022noise}, directly using polynomial interpolation to $R_{X}(\eta)$ induces at least exponentially increasing imprecision level with $N$, which essentially suppress the polynomial reduction from $|\text{GHE}|_{\pm}^{2}$ to ${\rm{GHED}}_{\pm}$.

To mitigate this exponential imprecision blowup induced by the polynomial interpolation, our key idea is to \textit{reduce} the degree of the polynomial $R_{X}(\eta)$.
More specifically, we construct a low-degree polynomial $R_{X}^{(l)}(\eta)$ by truncating the higher-order terms of $R_{X}(\eta)$, which still has the ideal value $P(X)$ at $\eta = 1$.
We find that the low-degree polynomial $R_{X}^{(l)}(\eta)$ can be sufficiently close to the original polynomial $R_{X}(\eta)$, which allows to construct polynomial interpolation for $R_{X}^{(l)}(\eta)$ given estimated values of $R_{X}(\eta)$. 
Here, we emphasize that for the polynomial reduction from $|\text{GHE}|_{\pm}^{2}$ to ${\rm{GHED}}_{\pm}$, the degree of the low-degree polynomial $R_{X}^{(l)}(\eta)$ should be at most $O(\log (\text{poly}(N, \epsilon_0^{-1}, \delta_0^{-1})))$ such that the polynomial interpolation process induces an imprecision blowup that scales at most $e^{O(\log (\text{poly}(N, \epsilon_0^{-1}, \delta_0^{-1})))} = \text{poly}(N, \epsilon_0^{-1}, \delta_0^{-1})$.


\subsubsection{Low-degree polynomial approximation for $R_{X}(\eta)$}\label{subsection: low degree aprroximation}


Before proceeding to the main proof, we derive the low-degree polynomial from $R_{X}(\eta)$ and establish the condition under which it remains sufficiently close to $R_{X}(\eta)$.

To obtain the low-degree polynomial, we conduct a series expansion of $R_{X}(\eta)$ in terms of $1-\eta$. 
More specifically, in Sec.~\ref{appendix:section:series}, we derive that $R_{X}(\eta)$ can be expanded as
\begin{align}\label{seriesexpansionofR}
R_{X}(\eta) = \sum_{n=0}^{N/2 - 1}(1-\eta)^{2n}M^{2n}\tanh^{2n}{r} \sum_{\substack{J \subseteq [N]\\|J|= N -2n}}\left|\Haf\left((XX^{T})_{J}\right)\right|^2 + (1-\eta)^{N}\tanh^{N}{r}M^N ,
\end{align}
where $(XX^T)_{J}$ is a $|J|$ by $|J|$ submatrix of $XX^T$ defined by taking rows and columns of $XX^T$ according to the subset $J \subseteq [N]$.

Given the series expansion in terms of $1-\eta$, we now see that the coefficient of the zeroth-order term $(1-\eta)^0$ is $\left|\Haf(XX^T)\right|^2 = P(X)$, i.e., what we are willing to estimate to solve $|\text{GHE}|_{\pm}^{2}$ problem. 
Accordingly, even if we truncate high-order terms in Eq.~\eqref{seriesexpansionofR}, the truncated polynomial still reduces to the desired value $P(X)$ at $\eta = 1$ as long as the zeroth-order term $(1-\eta)^0$ in Eq.~\eqref{seriesexpansionofR} remains unchanged.

Based on this understanding, we now lower the degree of the polynomial $R_{X}(\eta)$ by truncating high-order terms in Eq.~\eqref{seriesexpansionofR}.
Specifically, we define $R_{X}^{(l)}(\eta)$ as a truncated version of $R_{X}(\eta)$, only leaving up to $(1-\eta)^{l}$ term in the summation such that 
\begin{align}
R_{X}^{(l)}(\eta) &\coloneqq \sum_{n=0}^{l/2}(1-\eta)^{2n}M^{2n}\tanh^{2n}{r} \sum_{\substack{J \subseteq [N]\\|J|= N -2n}} \left| \Haf\left((XX^{T})_{J}\right) \right|^2 .
\end{align}
Indeed, $R_{X}^{(l)}(\eta)$ is a $l$-degree polynomial in $\eta$, and also contains the desired value at $\eta = 1$ such that $R_{X}^{(l)}(1) = P(X)$.

Next, we make sure that $R_{X}^{(l)}(\eta)$ is close enough to $R_{X}(\eta)$ for some $l$, because our strategy is to estimate the low-degree polynomial $R_{X}^{(l)}(\eta)$ instead of $R_{X}(\eta)$, given oracle access to ${\rm{GHED}}_{\pm}$.
However, since the distance between $R_{X}^{(l)}(\eta)$ and $R_{X}(\eta)$ increases with decreasing transmission rate $\eta$, the distance can be arbitrarily enlarged if $\eta$ is not bounded from below.
Hence, given the upper bound of the transmission rate $\eta \leq \eta^*$, we also lower bound the input transmission rate $\eta \geq \eta_{\rm{min}}$ for some yet-to-be-determined minimum input transmission rate $\eta_{\rm{min}} < \eta^*$, such that the input transmission rate $\eta$ is bounded in $[\eta_{\rm{min}},\eta^*]$.
Also, using our convention for the transmission rate $\eta$, the minimum $\eta_{\rm{min}}$ can also be represented as
\begin{align}\label{appendix: convention for eta-max}
\eta_{\rm{min}} \coloneqq 1 - \frac{k_{\rm{max}}}{N+k_{\rm{max}}},
\end{align}
where $k_{\rm{max}} > k^*$ now represents the maximum of the input $k$, such that $k \in [k^*, k_{\rm{max}}]$. 
We emphasize that $\eta^*$ ($k^*$) is the fixed, implicitly given parameter characterizing the noise rate of our lossy GBS, and $\eta_{\rm{min}}$ ($k_{\rm{max}}$) is the yet-to-be-specified parameters, which will be later specified in terms of the fixed parameter $\eta^*$ ($k^*$).

We now bound the distance between $R_{X}^{(l)}(\eta)$ and $R_{X}(\eta)$ from above, in terms of $l$ and $k_{\rm{max}}$.
Note that, the dependence on the maximum average number of lost photons $k_{\rm{max}}$ comes from the fact that the distance between $R_{X}^{(l)}(\eta)$ and $R_{X}(\eta)$ is maximized at $\eta = \eta_{\rm{min}}$. 
Specifically, we find that $R_{X}^{(l)}(\eta)$ and $R_{X}(\eta)$ is close enough over a large portion of $X \sim \mathcal{N}(0,1)_{\mathbb{C}}^{N \times M}$, for sufficiently large $l$ compared to $k_{\text{\rm{max}}}$.

\begin{lemma}\label{lemma:distancebytrunctionissmall}
Let $\sinh^2{r} = \frac{N}{\eta M}$, $\eta \ge \eta_{\rm{min}} = 1 - \frac{k_{\rm{max}}}{N + k_{\rm{max}}}$, and $l > k_{\rm{max}}$.
Then, for a matrix $X \sim \mathcal{N}(0,1)_{\mathbb{C}}^{N \times M}$ of i.i.d. Gaussians, the truncated polynomial $R_{X}^{(l)}(\eta)$ is close to $R_{X}(\eta)$ over a large fraction of $X$.
Specifically, the following inequality holds 
\begin{align}
\Pr_X\left[ |R_{X}(\eta) - R_{X}^{(l)}(\eta)| > \epsilon_1\binom{\frac{M}{2} + \frac{N}{2} - 1}{\frac{N}{2}}N!\right] < \delta,
\end{align}
for a truncation error $\epsilon_1$ given by
\begin{align}\label{epsilon1}
    \epsilon_1 = \frac{Nk_{\rm{max}}^{l}}{2\delta l!}.
\end{align}    
\end{lemma}
\begin{proof}
    See Sec.~\ref{appendix: section: proof of lemma}.
\end{proof}
Here, we emphasize that the truncation error $\epsilon_1$ scales $e^{-l\log l + l\log k_{\rm{max}} + O(l)}$ with $l$, which can be arbitrarily decreased by increasing the degree $l$ sufficiently larger than $k_{\rm{max}}$. 
Therefore, $R_{X}^{(l)}(\eta)$ can be made sufficiently close to $R_{X}(\eta)$ by setting $l$ sufficiently larger than $k_{\rm{max}}$.

To sum up, the reduction process from $|\text{GHE}|_{\pm}^{2}$ to ${\rm{GHED}}_{\pm}$ is as follows. 
By setting $l$ sufficiently larger than $k_{\text{\rm{max}}}$, the $l$-degree polynomial $R_{X}^{(l)}(\eta)$ can be well-estimated for $\eta \in [\eta_{\rm{min}},\eta^*]$ (or equivalently, $k \in [k^*, k_{\rm{max}}]$) via the oracle access to ${\rm{GHED}}_{\pm}$, i.e., by querying $P(\eta,X)$ value and multiplying $Q(\eta)$.
Then, using the well-estimated values of $R_{X}^{(l)}(\eta)$ corresponding to $l+1$ number of $\eta \le \eta^*$ values, we conduct polynomial interpolation to obtain the estimate of the polynomial $R_{X}^{(l)}(\eta)$. 
Using the estimate of the polynomial, we can infer the desired value $R_{X}^{(l)}(1) = P(X)$, which finally solves the $|\text{GHE}|_{\pm}^{2}$ problem.

\subsubsection{Reducing ${\text{GHED}}_{\pm}$ from $|\text{GHE}|_{\pm}^{2}$}

We are now ready to prove Theorem~\ref{maintheoreminformal}, i.e., establish reduction from $|\text{GHE}|_{\pm}^{2}$ to ${\rm{GHED}}_{\pm}$.
To proceed, we present a more general statement of Theorem~\ref{maintheoreminformal} as follows.

\begin{theorem}[A generalized version of Theorem~\ref{maintheoreminformal}]\label{maintheorem}
Let $\eta^*$ satisfy $(1-\eta^*) \leq \frac{1}{12\sqrt{N}}$.
Suppose $\mathcal{O}$ is an oracle that solves ${\rm GHED}_{\pm}$ with error bounds given by
\begin{align}
    \epsilon=O\left( \epsilon_0^3\delta_0^2N^{-3}e^{-\alpha k^*} \right), \qquad \delta = O\left(\frac{\delta_0}{k^* + \log (N\epsilon_0^{-1}\delta_0^{-1})} \right),
\end{align}
for $k^* = (\frac{1}{\eta^*} - 1)N$ and $\alpha =6(1+\log2.1)e^2\chi + 3$ with
\begin{align}
    \chi = \exp[\frac{\log\log(N\epsilon_0^{-1}\delta_0^{-1})}{\log(N\epsilon_0^{-1}\delta_0^{-1})}].
\end{align}
Then $|\rm{GHE}|_{\pm}^{2}$ can be solved in $\rm{BPP}^{\mathcal{O}}$. 
\end{theorem}

Before the proof, let us remark on how Theorem~\ref{maintheorem} leads to Theorem~\ref{maintheoreminformal}, i.e., the classical hardness result of lossy GBS for the logarithmic number of lost photons. 
More specifically, let $k^*$ be upper bounded by $k^* \leq k_{\rm{th}}$ for some threshold value $k_{\rm{th}} = O(\log N)$.
Then, one can easily find that $\epsilon$ and $\delta$ in Theorem~\ref{maintheorem} are polynomially related to $N^{-1}$, $\epsilon_0$ and $\delta_0$ (i.e., $\epsilon, \delta = O(\text{poly}(N,\epsilon_0^{-1},\delta_0^{-1})^{-1})$).
Therefore, as long as $k_{\rm{th}} = O(\log N)$, ${\rm{GHED}}_{\pm}$ is polynomially reducible from $|\text{GHE}|_{\pm}^{2}$, and accordingly, Theorem~\ref{maintheoreminformal} is straightforward.

\begin{proof}[Proof of Theorem~\ref{maintheorem}]
Let $\mathcal{O}$ be the oracle that solves ${\rm{GHED}}_{\pm}$ problem, i.e., on input $\eta \in [0, \eta^*]$ and $X \sim \mathcal{N}(0,1)_{\mathbb{C}}^{N\times M}$, estimates $P(\eta,X)$ to within $\epsilon \binom{\frac{M}{2} + \frac{N}{2} - 1}{\frac{N}{2}}N!$ over $1-\delta$ of $X$.
More concretely, we have
\begin{equation}
    \Pr_{X}\left[\frac{|\mathcal{O}(\eta,X) - P(\eta,X)|}{ \binom{\frac{M}{2} + \frac{N}{2} - 1}{\frac{N}{2}}N!} > \epsilon \right] < \delta ,
\end{equation}
where we divided both sides in $\Pr[\cdot]$ by the multiplicative factor $\binom{\frac{M}{2} + \frac{N}{2} - 1}{\frac{N}{2}}N!$ for more clarity. 
Because $P(\eta,X) = Q(\eta)^{-1}R_{X}(\eta)$, one can estimate the polynomial $R_{X}(\eta)$ by multiplying $Q(\eta)$ on given estimated value of $P(\eta,X)$.
We note that this process can be done in classical polynomial time; we find in Sec.~\ref{appendix:section:efficient} that $Q(\eta)$ is efficiently computable to within a desired imprecision level.

To bound the error induced by the multiplicative factor $Q(\eta)$, we denote $Q_{\rm{max}}$ as the maximum value of $Q(\eta)$ over $\eta \in [\eta_{\rm{min}}, \eta^*]$ for some yet-to-be-determined free parameter $\eta_{\rm{min}} < \eta^*$, such that
\begin{align}\label{eq: def of Qmax}
    Q_{\rm{max}} \coloneqq \max_{\eta\in [\eta_{\rm{min}},\eta^*]}Q(\eta).
\end{align}
Then, for any input noise rate $\eta \in [\eta_{\rm{min}},\eta^*]$, one can estimate $R_{X}(\eta)$ with high probability over $X$ as
\begin{align}
    \Pr_{X}\left[\frac{|\mathcal{O}(\eta,X)Q(\eta) - R_{X}(\eta)|}{\binom{\frac{M}{2} + \frac{N}{2} - 1}{\frac{N}{2}}N!} > Q_{\rm{max}}\epsilon  \right] &\le \Pr_{X}\left[\frac{|\mathcal{O}(\eta,X)Q(\eta) - R_{X}(\eta)|}{\binom{\frac{M}{2} + \frac{N}{2} - 1}{\frac{N}{2}}N!} > Q(\eta)\epsilon  \right] \\
    &= \Pr_{X}\left[\frac{|\mathcal{O}(\eta,X) - P(\eta,X)|}{\binom{\frac{M}{2} + \frac{N}{2} - 1}{\frac{N}{2}}N!} > \epsilon  \right] \\
    &< \delta  \label{eeee}.   
\end{align}

Next, by Lemma~\ref{lemma:distancebytrunctionissmall}, the difference between $R_{X}(\eta)$ and $R_{X}^{(l)}(\eta)$ is at most $\epsilon_1\binom{\frac{M}{2} + \frac{N}{2} - 1}{\frac{N}{2}}N!$ with high probability over $X \sim \mathcal{N}(0,1)_{\mathbb{C}}^{N \times M}$, where the truncation error $\epsilon_1$ is defined in Eq.~\eqref{epsilon1}.
Using this fact, one can also well-estimate $R_{X}^{(l)}(\eta)$ by estimating $R_{X}(\eta)$ for the average-case $X$. 
To proceed, let us denote $\epsilon' = Q_{\rm{max}}\epsilon + \epsilon_1$. 
Then we have
\begin{align}
\Pr_{X}\left[\frac{|\mathcal{O}(\eta,X)Q(\eta) - R_{X}^{(l)}(\eta)|}{\binom{\frac{M}{2} + \frac{N}{2} - 1}{\frac{N}{2}}N!} > \epsilon' \right] 
&\le \Pr_{X}\left[\frac{|\mathcal{O}(\eta,X)Q(\eta) - R_{X}(\eta)|}{ \binom{\frac{M}{2} + \frac{N}{2} - 1}{\frac{N}{2}}N!} > Q_{\rm{max}}\epsilon \right]  + \Pr_X\left[ \frac{|R_{X}(\eta) - R_{X}^{(l)}(\eta)|}{ \binom{\frac{M}{2} + \frac{N}{2} - 1}{\frac{N}{2}}N!} > \epsilon_1\right] \label{qqqq} \\
&< 2\delta \label{wwww} ,
\end{align}
where we used triangular inequality in Eq.~\eqref{qqqq}, and used Lemma~\ref{lemma:distancebytrunctionissmall} and Eq.~\eqref{eeee} in Eq.~\eqref{wwww}.

Accordingly, one can estimate $R_{X}^{(l)}(\eta)$ to within $\epsilon'\binom{\frac{M}{2} + \frac{N}{2} - 1}{\frac{N}{2}}N!$ with probability at least $1 - 2\delta$ over $X \sim \mathcal{N}(0,1)_{\mathbb{C}}^{N\times M}$, given an oracle access to ${\rm{GHED}}_{\pm}$ and multiplying the computed value $Q(\eta)$. 
Based on this understanding, we infer the desired value $R_{X}^{(l)}(1)$ from the estimated values of $R_{X}^{(l)}(\eta)$ for $\eta \in [\eta_{\rm{min}}, \eta^*]$ via polynomial interpolation.

For convenience of later polynomial interpolation, using the conventions for $\eta^*$ and $\eta_{\rm{min}}$ in terms of $k^*$ and $k_{\rm{max}}$ depicted in Eq.~\eqref{appendix: convention for eta-star} and Eq.~\eqref{appendix: convention for eta-max}, we replace the $l$-degree polynomial to $F(x) = R_{X}^{(l)}(g(x))$ with a linear function 
\begin{align}
    g(x) = \frac{N(k_{\rm{max}} + k^*) + 2k_{\rm{max}}k^*}{2(N+k_{\rm{max}})(N+k^*)}(x-1) + 1 ,
\end{align}
to rescale the input variable $\eta$ to $x\in [-\Delta, \Delta]$ for some parameter $\Delta$ (not specified yet) while maintaining the degree of the polynomial to $l$.
Note that by definition, $g(1) = 1$, and thus we still have the desired value at $x=1$, i.e., $F(1) = R_{X}^{(l)}(1) = P(X)$.

To ensure that $g(x) = \eta$ is bounded in $[\eta_{\rm{min}}, \eta^*]$ for $x \in [-\Delta, \Delta]$, $\Delta$ must satisfies $g(-\Delta) \geq \eta_{\rm{min}}$ and $g(\Delta) \leq \eta^*$ at each boundary (because $g(x)$ monotonically increases with $x$), which leads to the condition
\begin{align}\label{appendix: eq: upper bound of Delta}
    \Delta \leq \frac{1}{1 + \frac{2k_{\rm{max}}k^*}{N}}\frac{k_{\rm{max}} - k^*}{k_{\rm{max}} + k^*}.
\end{align}
Since we have the freedom to choose the value of $k_{\rm{max}}$ for given $k^*$, at this moment we do not specify $k_{\rm{max}}$ except that $k_{\rm{max}}$ identically scales with $k^*$ such that $\frac{k_{\rm{max}}}{k^*} = \Theta(1)$; later, we will specify the value of $k_{\rm{max}}$ in terms of $k^*$ that satisfies this condition. 
To proceed, we further assume that $k_{\rm{max}}$ and $k^*$ satisfy $\frac{k_{\rm{max}}k^*}{N} \leq \frac{1}{40}$, which automatically implies $k_{\rm{max}},\,k^* \leq O(\sqrt{N})$.
Under this assumption, setting $\Delta = \frac{20}{21}\frac{k_{\rm{max}}-k^*}{k_{\rm{max}} + k^*}$ is sufficient to satisfy the condition in Eq.~\eqref{appendix: eq: upper bound of Delta}. 
Note that this also implies $\Delta = \Theta(1)$ with $\Delta \in (0,1)$.


We now prepare the $l + 1$ number of input values $\{x_i\}_{i=1}^{l+1}$, which is a set of equally spaced points in the interval $x_i \in [-\Delta, \Delta]$.
For each $x_i$, we obtain a estimation value $y_i =  \mathcal{O}(g(x_i),X)Q(g(x_i))$ for $F(x_i)$ which satisfies the bound in Eq.~\eqref{wwww} such that
\begin{equation}
    \Pr\left[|y_i - F(x_i)| > \epsilon'\binom{\frac{M}{2} + \frac{N}{2} - 1}{\frac{N}{2}}N! \right] < 2\delta.
\end{equation}
Using those estimation values, we infer the value $F(1) = R_{X}^{(l)}(1)  = P(X)$;
to do so, we use the Lagrange interpolation for $F(x)$, under the condition that all the $(x_i, y_i)$ points satisfy $|y_i - F(x_i)| \le \epsilon'\binom{\frac{M}{2} + \frac{N}{2} - 1}{\frac{N}{2}}N!$.
By using the union bound, the probability that all the $l+1$ number of $y_i$ points are $\epsilon'\binom{\frac{M}{2} + \frac{N}{2} - 1}{\frac{N}{2}}N!$-close to $F(x_i)$ is at least $(1-2\delta)^{l+1}$. 
Given that all the $l+1$ points are successful, the Lagrange interpolation gives an estimation value $F(1)$, where we use the error bound of Lagrange interpolation derived by Ref.~\cite{kondo2022quantum} as follows.  

\begin{lemma}[Kondo et al~\cite{kondo2022quantum}]
Let $h(x)$ be a polynomial of degree at most $d$, Let $\Delta\in(0,1)$. Assume that $|h(x_j)|\le \epsilon$ for all of the $d+1$ equally-spaced points $x_j = -\Delta + \frac{2j}{d}\Delta$ for $j = 0,\dots,d$. Then
\begin{equation}\label{interpolationerror}
    |h(1)| < \epsilon\frac{\exp[d(1+\log\Delta^{-1})]}{\sqrt{2\pi d}}.
\end{equation}
\end{lemma}

By drawing $l$-degree polynomial $F_e(x)$ with $\{(x_i,y_i)\}_{i=1}^{l+1}$ points, for $h(x) \coloneq F_e(x) - F(x)$, $|h(x_i)| \le \epsilon'\binom{\frac{M}{2} + \frac{N}{2} - 1}{\frac{N}{2}}N!$ for $x_i$ satisfying $|y_i - F(x_i)| \le \epsilon'\binom{\frac{M}{2} + \frac{N}{2} - 1}{\frac{N}{2}}N!$.
Therefore, given the error bound in Eq.~\eqref{interpolationerror}, we can obtain an estimator $F_e(1)$ for $F(1)$ whose error is bounded from above as $|h(1)| < \epsilon'\frac{\exp[l(1+\log\Delta^{-1})]}{\sqrt{2\pi l}}\binom{\frac{M}{2} + \frac{N}{2} - 1}{\frac{N}{2}}N! < \epsilon'e^{l(1+\log\Delta^{-1})}\binom{\frac{M}{2} + \frac{N}{2} - 1}{\frac{N}{2}}N!$ via Lagrange interpolation, with success probability at least $(1-2\delta)^{l+1}$. 
More concretely, we can estimate the desired value $F(1) = P(X)$ satisfying
\begin{align}\label{estimationbyinterpolation}
\begin{split}
    \Pr\left[|F_e(1) - P(X)| > \epsilon'e^{l(1+\log\Delta^{-1})}\binom{\frac{M}{2} + \frac{N}{2} - 1}{\frac{N}{2}}N! \right] 
    < 1 - (1-2\delta)^{l+1}.
\end{split}
\end{align}

Therefore, to solve $|\rm{GHE}|_{\pm}^{2}$ problem, it is sufficient to set the error parameters $\epsilon$ and $\delta$ satisfying $\epsilon_0 \ge \epsilon'e^{l(1+\log\Delta^{-1})}$ and $\delta_0 \ge 1 - (1-2\delta)^{l+1}$.
Accordingly, we set $\delta$ satisfying $\delta_0 = 1 - (1-2\delta)^{l+1}$, which implies that $\delta = O(\delta_0 l^{-1})$.
Also, by combining the results so far, we have
\begin{align}
    \epsilon_0 \ge e^{l(1+\log\Delta^{-1})}(Q_{\rm{max}}\epsilon + \epsilon_1),
\end{align}
which requires $\epsilon$ to satisfy the following condition:
\begin{align}\label{conditionforreduction}
\epsilon \le (\epsilon_0 e^{-l(1+\log\Delta^{-1})} - \epsilon_1)Q_{\rm{max}}^{-1},
\end{align}
where the truncation error $\epsilon_1$ is given in Eq.~\eqref{epsilon1}.

However, since we set the error parameter $\epsilon$ in ${\rm{GHED}}_{\pm}$ problem positive, one needs to make $\epsilon_0 e^{-l(1+\log\Delta^{-1})}$ sufficiently larger than $\epsilon_1$ such that the condition reduces to $\epsilon \leq O(\epsilon_0 e^{-l(1+\log\Delta^{-1})}Q_{\rm{max}}^{-1})$. 
On the one hand, since the truncation error $\epsilon_1$ given in Eq.~\eqref{epsilon1} scales $\sim k_{\rm{max}}^{l}/l! = e^{-l\log l + l\log k_{\rm{max}} + O(l)}$, one can arbitrarily increase $l$ to make $\epsilon$ sufficiently smaller than $\epsilon_0 e^{-l(1+\log\Delta^{-1})}$.
On the other hand, arbitrarily increasing $l$ makes arbitrarily small $\epsilon$ by the condition $\epsilon \leq O(\epsilon_0 e^{-l(1+\log\Delta^{-1})}Q_{\rm{max}}^{-1})$ for the reduction.
Therefore, it is necessary to find an intermediate-sized $l$ to make $\epsilon_1 \ll \epsilon_0 e^{-l(1+\log\Delta^{-1})}$ and also prevent arbitrarily decreasing $\epsilon_0 e^{-l(1+\log\Delta^{-1})}Q_{\rm{max}}^{-1}$.

We derive in Sec.~\ref{appendix: section: investigation} that by setting $l$ in terms of $k_{\rm{max}}$ as
\begin{align}
    l = \frac{e^2}{\Delta}\chi k_{\rm{max}}  + \log(N \epsilon_0^{-1}\delta_0^{-1}),
\end{align}
for $\chi$ given by
\begin{align}\label{appendix: eq: chi}
    \chi = \exp[\frac{\log\log(N\epsilon_0^{-1}\delta_0^{-1})}{\log(N\epsilon_0^{-1}\delta_0^{-1})}],
\end{align}
it is sufficient to make $\epsilon_1 = o(\epsilon_0 e^{-l(1+\log\Delta^{-1})})$, such that the condition in Eq.~\eqref{conditionforreduction} reduces to 
\begin{align}
    \epsilon \leq O\left(\epsilon_0 e^{-(1+\log\Delta^{-1}) \left[\frac{e^2}{\Delta}\chi k_{\rm{max}}  + \log(N \epsilon_0^{-1}\delta_0^{-1}) \right]}Q_{\rm{max}}^{-1} \right).
\end{align}
Also, we find in Sec.~\ref{appendix: section: investigation} that $Q_{\rm{max}}$ satisfies $Q_{\rm{max}}^{-1} \geq \Omega(\frac{1}{\sqrt{N}}e^{-k_{\rm{max}}})$ provided that $k_{\rm{max}} \leq O(\sqrt{N})$.
Lastly, to set the error bound in terms of given $k^*$, we now specify the maximum value of $k$ as $k_{\rm{max}} = 3 k^*$, which makes $\Delta = \frac{20}{21}\frac{k_{\rm{max}} - k^*}{k_{\rm{max}} + k^*} = \frac{10}{21}$.
Combining all these arguments, we finally obtain the condition for $\epsilon$ given by
\begin{align}
\epsilon \le O\left( \frac{\epsilon_0^{2+\log 2.1}\delta_0^{1+\log 2.1}}{N^{\frac{3}{2} +\log2.1 }} e^{-\left(6(1+\log2.1)e^2\chi +3 \right)k^*  }\right).
\end{align}
Additionally, since we assumed that $\frac{k_{\rm{max}}k^*}{N} \leq \frac{1}{40}$, our setting $k_{\rm{max}} = 3 k^*$ yields the condition $\frac{(k^*)^2}{N} \leq \frac{1}{120}$, which is satisfied given that $1 - \eta^* \leq \frac{1}{12\sqrt{N}}$ as stated in Theorem~\ref{maintheorem}.

To sum up, by setting $\delta$ and $\epsilon$ as
\begin{align}
    \delta &= O\left(\delta_0l^{-1}\right) \\
    &= O\left(\frac{\delta_0}{k^* + \log (N\epsilon_0^{-1}\delta_0^{-1})} \right),
\end{align}
and
\begin{align}
\epsilon = O\left( \epsilon_0^{3}\delta_0^{2}N^{-3} e^{-\alpha k^*  }\right) ,
\end{align}
for $\alpha \coloneqq 6(1+\log2.1)e^2\chi + 3$ with $\chi$ given in Eq.~\eqref{appendix: eq: chi}, one can estimate $P(X)$ to within $\epsilon_0 \binom{\frac{M}{2} + \frac{N}{2} - 1}{\frac{N}{2}}N!$ with probability at least $1 - \delta_0$ over $X$, given access to the oracle $\mathcal{O}$ for the ${\rm{GHED}}_{\pm}$ problem.
This completes the proof.
\end{proof}

\subsubsection{Specifying parameters for the reduction}\label{appendix: section: investigation}

In the remainder, we specify the parameters $k_{\text{\rm{max}}}$ and $l$, first to make the truncation error $\epsilon_1 = \frac{Nk_{\rm{max}}^{l}}{2\delta l!}$ scale smaller than $ \epsilon_0 e^{-l(1+\log\Delta^{-1})}$ such that $\epsilon_1 = o(\epsilon_0 e^{-l(1+\log\Delta^{-1})}) $, while avoiding $\epsilon_0 e^{-l(1+\log\Delta^{-1})}Q_{\rm{max}}^{-1}$ becomes exceedingly small. 
To do so, we aim to find a minimum $l$ that satisfies the first condition.

To satisfy the first condition, we first figure out $l$ that makes the following quantity $o(1)$:
\begin{align}
\frac{\epsilon_1}{\epsilon_0 e^{-l(1+\log\Delta^{-1})}} 
&=  \frac{\epsilon_0^{-1}\delta^{-1}}{2}\frac{Nk_{\rm{max}}^{l}}{l!}e^{l(1+\log\Delta^{-1})} \\
&\leq \frac{N\epsilon_0^{-1}\delta^{-1}}{2\sqrt{2\pi l}}e^{l\log k_{\rm{max}} - l\log l+2l+l\log\Delta^{-1}} \label{asdfqwerzx}\\
&= O\left( e^{\log\sqrt{l} + l\log\left(\frac{k_{\rm{max}}e^2}{\Delta l}\right) + \log (N\epsilon_0^{-1}\delta_0^{-1}) }\right) , \label{qwaszx}
\end{align}
where we used Stirling's inequality for $l!$ in Eq.~\eqref{asdfqwerzx}, and used the relation $\delta_0 = 1 - (1-2\delta)^{l+1}$ in Eq.~\eqref{qwaszx} which implies that $\delta^{-1} = O(l\delta_0^{-1})$.
Note that as we increase $l$, the right-hand side of Eq.~\eqref{qwaszx} gets smaller, and it will become $o(1)$ at some point.

Our approach is first to parameterize $l$ in terms of $k_{\rm{max}}$ as $l = \Gamma k_{\rm{max}}$ for a yet-to-be-determined constant $\Gamma$.
We further parameterize $\log(N\epsilon_0^{-1}\delta_0^{-1})$ in terms of $k_{\rm{max}}$ as $\log(N\epsilon_0^{-1}\delta_0^{-1}) = \Lambda k_{\rm{max}}$ for $\Lambda \in [0,\infty)$, such that $\Lambda$ is determined by the relation between $\log(N\epsilon_0^{-1}\delta_0^{-1})$ and $k_{\rm{max}}$. For example, if $k_{\rm{max}} = o(\log(N\epsilon_0^{-1}\delta_0^{-1}))$, then $\Lambda$ will diverge to $\infty$ as system size scales.
On the other hand, if $k_{\rm{max}} = \omega(\log(N\epsilon_0^{-1}\delta_0^{-1}))$, then $\Lambda$ will converge to $0$.

Using these conventions, in the exponent on the right-hand side of Eq.~\eqref{qwaszx}, we have
\begin{align}\label{appendix: eq: exponent}
    \log\sqrt{\Gamma k_{\rm{max}}} + \Gamma k_{\rm{max}}\log(\frac{e^2}{\Delta\Gamma}) + \Lambda k_{\rm{max}} = \log\sqrt{\Gamma k_{\rm{max}}} + \left[\Gamma\left( \log(\frac{e^2}{\Delta}) - \log\Gamma \right) + \Lambda \right]k_{\rm{max}} .
\end{align}
To make the right-hand side of Eq.~\eqref{qwaszx} scale $o(1)$, the right-hand side term in Eq.~\eqref{appendix: eq: exponent} should be at least negative, and this requires at least $\Gamma > \frac{e^2}{\Delta}$.
Accordingly, we set $\Gamma$ as 
\begin{align}
    \Gamma = \frac{e^2}{\Delta}\chi + \Lambda ,
\end{align}
for some parameter $\chi \geq 1$ not specified at this moment. 
Then, from the right-hand side of Eq.~\eqref{appendix: eq: exponent}, we have 
\begin{align}
   \Gamma\left( \log(\frac{e^2}{\Delta}) - \log\Gamma \right) + \Lambda &= -\left( \frac{e^2}{\Delta}\chi + \Lambda\right) \left( \log\chi + \log(1 + \frac{\Delta}{e^2\chi}\Lambda) \right) + \Lambda \\
    &= -\left( \frac{e^2}{\Delta}\chi + \Lambda\right) \log\chi -  \left( \frac{e^2}{\Delta}\chi + \Lambda\right)  \log(1 + \frac{\Delta}{e^2\chi}\Lambda)  + \Lambda \label{appendix: eq: ss} \\
    &\leq -\left( \frac{e^2}{\Delta}\chi + \Lambda \right)\log\chi, \label{appendix: eq: sa}
\end{align}
where the inequality in Eq.~\eqref{appendix: eq: sa} holds because the function $-(a+x)\log(1+x/a) + x \leq 0$ for all $x\geq 0$ given that $a > 0$.
Hence, we obtain 
\begin{align}
    \log\sqrt{\Gamma k_{\rm{max}}} + \Gamma k_{\rm{max}}\log(\frac{e^2}{\Delta\Gamma}) + \Lambda k_{\rm{max}} &\leq \frac{1}{2}\log(\left( \frac{e^2}{\Delta}\chi + \Lambda\right) k_{\rm{max}})-\left( \frac{e^2}{\Delta}\chi + \Lambda \right)(\log\chi) k_{\rm{max}} \\
    &\leq -\frac{1}{2}\log(\left( \frac{e^2}{\Delta}\chi + \Lambda\right) k_{\rm{max}}) \\
    & = -\frac{1}{2}\log( \frac{e^2}{\Delta}\chi k_{\rm{max}} + \log(N\epsilon_0^{-1}\delta_0^{-1})) ,
\end{align}
where the second inequality in the above holds as long as $\chi$ satisfies 
\begin{align}\label{appendix: eq: condition for chi}
    \log \chi \geq  \frac{\log( \frac{e^2}{\Delta}\chi k_{\rm{max}} + \log(N\epsilon_0^{-1}\delta_0^{-1}))}{\frac{e^2}{\Delta}\chi k_{\rm{max}} + \log(N\epsilon_0^{-1}\delta_0^{-1})} .
\end{align}
Note that $\log(x)/x$ is maximized at $x = e$ and then monotonically decreases to 0 as $x \rightarrow \infty$.
Accordingly, we define $\chi$ as
\begin{align}\label{appendix: eq: def of chi}
    \chi = \exp[\frac{\log\log(N\epsilon_0^{-1}\delta_0^{-1})}{\log(N\epsilon_0^{-1}\delta_0^{-1})}],
\end{align}
which satisfies the condition in Eq.~\eqref{appendix: eq: condition for chi} as long as the system size is not too small, i.e., $\log(N\epsilon_0^{-1}\delta_0^{-1}) \ge e$.
Combining all the above arguments, starting from Eq.~\eqref{qwaszx}, we finally obtain the bound
\begin{align}
    \frac{\epsilon_1}{\epsilon_0 e^{-l(1+\log\Delta^{-1})}} 
     \leq O\left( \frac{1}{\sqrt{\frac{e^2}{\Delta}\chi k_{\rm{max}} + \log(N\epsilon_0^{-1}\delta_0^{-1})}} \right),
\end{align}
where the right-hand side remains $o(1)$ for any $k_{\rm{max}}$, thus ensuring that the first condition is satisfied.

To sum up, by setting $l = \Gamma k_{\rm{max}} = \frac{e^2}{\Delta}\chi k_{\rm{max}} + \log(N\epsilon_0^{-1}\delta_0^{-1})$ for $\chi$ given in Eq.~\eqref{appendix: eq: def of chi}, we can make $\epsilon_1 = o(\epsilon_0 e^{-l(1+\log\Delta^{-1})})$.
Now, the condition for $\epsilon$ reduces to 
\begin{align}\label{appendix: eq: condition for epsilon}
    \epsilon &\leq O\left(\epsilon_0 e^{-l(1+\log\Delta^{-1})}Q_{\rm{max}}^{-1} \right)  \\
    &= O\left( \epsilon_0 e^{-(1+\log\Delta^{-1})\frac{e^2}{\Delta}\chi k_{\rm{max}}   -(1+\log\Delta^{-1}) \log(N\epsilon_0^{-1}\delta_0^{-1})}Q_{\rm{max}}^{-1}\right).
\end{align}
To proceed, we find in Sec.~\ref{appendix:section:efficient} that $Q(\eta) \le O(\sqrt{N}e^{(1-\eta)N})$ provided that $(1-\eta) \leq O(\frac{1}{\sqrt{N}})$.
Since we assume that $(1-\eta_{\rm{min}})$ scales identically to $(1-\eta^*) \leq \frac{1}{12\sqrt{N}}$, it follows that $(1-\eta_{\rm{min}}) \leq O(\frac{1}{\sqrt{N}})$.
Accordingly, $Q_{\rm{max}} \le O(\sqrt{N}e^{(1-\eta_{\rm{min}})N})$ and thus $Q_{\rm{max}}^{-1} \ge \Omega(\frac{1}{\sqrt{N}}e^{-(1-\eta_{\rm{min}})N})$.
To describe this lower bound in terms of $k_{\rm{max}}$, we have
\begin{align}
    e^{-(1-\eta_{\rm{min}})N} = e^{-\left(\frac{1}{\eta_{\rm{min}}} - 1\right)N + \frac{(1-\eta_{\rm{min}})^2}{\eta_{\rm{min}}}N} = \Theta(e^{-k_{\rm{max}}}),
\end{align}
provided that $(1-\eta_{\rm{min}}) \leq O(\frac{1}{\sqrt{N}})$.
Hence, given the lower bound of $Q_{\rm{max}}^{-1} \ge \Omega(\frac{1}{\sqrt{N}}e^{-k_{\rm{max}}})$, we have the condition of $\epsilon$ from Eq.~\eqref{appendix: eq: condition for epsilon} as
\begin{align}\label{modifiedconditionforepsilon}
\epsilon \le O\left(\frac{\epsilon_0 }{\sqrt{N}} e^{-\left((1+\log\Delta^{-1})\frac{e^2}{\Delta}\chi + 1 \right)k_{\rm{max}}   -(1+\log\Delta^{-1}) \log(N\epsilon_0^{-1}\delta_0^{-1})}\right).
\end{align}

Finally, as we have the freedom to choose the value of $\eta_{\rm{min}}$, now we specify the value of $k_{\rm{max}}$ as $k_{\rm{max}} = 3k^*$, which leads to $\Delta = \frac{20}{21}\frac{k_{\rm{max}} - k^*}{k_{\rm{max}} + k^*} = \frac{10}{21}$.
Then, the condition in Eq.~\eqref{modifiedconditionforepsilon} reduce to
\begin{align}
\epsilon &\le O\left(\frac{\epsilon_0 }{\sqrt{N}} e^{-3\left(2(1+\log2.1)e^2\chi + 1 \right)k^*   -(1+\log 2.1) \log(N\epsilon_0^{-1}\delta_0^{-1})}\right) \\
& = O\left( \epsilon_0^{2+\log 2.1}\delta_0^{1+\log 2.1}N^{-\frac{3}{2}-\log2.1 } e^{-3\left(2(1+\log2.1)e^2\chi + 1 \right)k^*  }\right),
\end{align}
thus obtaining the desired bound.

\section{Revisiting hardness proof of GBS}\label{supplement: hardness proof of GBS}

This section recaps the current hardness argument for ideal GBS and discusses how to extend this result for the lossy GBS case.

The current proof technique for the hardness of sampling problems, including GBS, essentially builds upon Stockmeyer's algorithm about approximate counting~\cite{stockmeyer1985approximation}.
Specifically, given a randomized classical sampler that outputs a sample from a given output distribution, Stockmeyer's algorithm enables one to multiplicatively estimate a fixed output probability of the sampler, within complexity class $\rm{BPP}^{\rm{NP}}$. 
Hence, if there exists an efficient classical sampler that mimics GBS, then one can approximate the output probability of GBS within complexity class $\rm{BPP}^{\rm{NP}}$, where the approximation error is determined by how well such a classical sampler can mimic GBS.

More formally, suppose there exists a randomized classical sampler $\mathcal{S}_0$ that, on input a circuit unitary matrix $U$, samples the $N$-photon outcome according to the output distribution of (post-selected) ideal GBS (i.e., $p_{S}(U)$ in Eq.~\eqref{supple: ideal output probability}), within total variation distance $\beta$. 
Then, for $U$ given as random unitary matrix drawn from Haar measure, one can prove that given access to the sampler $\mathcal{S}_0$, $|\text{GHE}|_{\pm}^{2}$ defined in Problem~\ref{def: GHE} can be solved for $\beta = O(\epsilon_0\delta_0)$ using Stockmeyer's algorithm~\cite{stockmeyer1985approximation}, such that $|\text{GHE}|_{\pm}^{2} \in \rm{BPP}^{NP^{\mathcal{S}_0}}$; we leave a formal proof of this at the end of this section. 
Hence, if an efficient classical algorithm for $\mathcal{S}_0$ with $\beta = O(\epsilon_0\delta_0)$ exists, $|\text{GHE}|_{\pm}^{2}$ can be solved in $\text{BPP}^{\text{NP}}$, and thus \#P-hardness of $|\text{GHE}|_{\pm}^{2}$ implies the collapsion of polynomial hierarchy to the finite level $\text{BPP}^{\text{NP}}$.
In other words, proving \#P-hardness of $|\text{GHE}|_{\pm}^{2}$ implies the non-existence of the classical algorithm that can simulate the ideal GBS within the total variation distance $\beta$ in $\text{poly}(N,\beta^{-1})$ time, given that polynomial hierarchy does not collapse to the finite level.
More detailed analysis can be checked in Ref.~\cite{aaronson2011computational, kruse2019detailed}.

We can naturally extend this argument to the lossy GBS case by replacing the problem $|\text{GHE}|_{\pm}^{2}$ to its noisy version ${\rm{GHED}}_{\pm}$ defined in Problem~\ref{def: sumGHE}.
Specifically, consider there exists a randomized classical sampler $\mathcal{S}$ that, on input a circuit unitary matrix $U$ and a transmission rate $\eta \in [0, \eta^*]$, samples the $N$-photon outcome according to the output distribution of (post-selected) lossy GBS (i.e., $p_{S}(\eta, U)$ in Eq.~\eqref{supple: noisy output probability}), within total variation distance error $\beta$.
Then, for Haar-random unitary matrix $U$, exactly following the analysis for the ideal case, one can find that ${\rm{GHED}}_{\pm} \in \rm{BPP}^{NP^{\mathcal{S}}}$.
Hence, if an efficient classical algorithm for the lossy GBS sampler $\mathcal{S}$ with $\eta \in [0, \eta^*]$ and $\beta = O(\epsilon\delta)$ exists, ${\rm{GHED}}_{\pm}$ can be solved in $\text{BPP}^{\text{NP}}$.
Accordingly, similarly as before, proving \#P-hardness of ${\rm{GHED}}_{\pm}$ implies that no classical algorithm can simulate the lossy GBS within the total variation distance $\beta$ in $\text{poly}(N,\beta^{-1})$ time, unless polynomial hierarchy does not collapse to the finite level.

For a more self-contained analysis, we give a proof of $|\text{GHE}|_{\pm}^{2} \in \rm{BPP}^{NP^{\mathcal{S}_0}}$ in the following.
This essentially follows the analysis given in Ref.~\cite{aaronson2011computational, kruse2019detailed} but using slightly different scaling factors.
We emphasize that the proof of ${\rm{GHED}}_{\pm} \in \rm{BPP}^{NP^{\mathcal{S}}}$ is straightforward from the proof of $|\text{GHE}|_{\pm}^{2} \in \rm{BPP}^{NP^{\mathcal{S}_0}}$;
we just need to replace the ideal output probability with the noisy output probability in the proof.

\begin{proof}[Proof of $|\rm{GHE}|_{\pm}^{2} \in \rm{BPP}^{NP^{\mathcal{S}_0}}$]

To solve $|\text{GHE}|_{\pm}^{2}$ problem, given a Gaussian random matrix $X \sim \mathcal{N}(0,1)_{\mathbb{C}}^{N\times M}$ and error parameters $\epsilon_0$, $\delta >0$, one needs to estimate $P(X)$ to within additive error $\epsilon_0\binom{\frac{M}{2} + \frac{N}{2} - 1}{\frac{N}{2}}N!$ with probability at least $1-\delta$, in $\text{poly}(N, \epsilon_0^{-1},\delta_0^{-1})$ time. 
We now show that this task can be solved in complexity class $\text{BPP}^{\text{NP}}$, given access to the randomized classical sampler $\mathcal{S}_0$.

Before proceeding, we first denote $\mathcal{H}^{M}$ as a distribution of the $M$ by $M$ unitary matrix drawn from the Haar measure on U(M). 
Because we are considering $M \propto N^{\gamma}$ with large enough $\gamma > 2$, the distribution of $XX^{T}$ has bounded total variation distance to the distribution of $M(UU^T)_{S}$ for $M$ by $M$ Haar-random unitary matrix $U \sim \mathcal{H}^{M}$ and any collision-free outcome $S$~\cite{jiang2009entries, shou2025proof}.
Then, as argued in~\cite{aaronson2011computational, kruse2019detailed}, there exists $\text{BPP}^{\text{NP}}$ procedure that samples $M$ by $M$ Haar-random unitary matrix $U \sim \mathcal{H}^{M}$ with high probability such that given $X \sim \mathcal{N}(0,1)_{\mathbb{C}}^{N\times M}$, $\frac{1}{M}(XX^T)$ is exactly an $N$ by $N$ random submatrix of $UU^T$.
Moreover, as the distribution similarity between $XX^{T}$ and $M(UU^T)_{S}$ increases with $M$~\cite{shou2025proof}, the success probability can be bounded to at least $1 - \frac{\delta_0}{4}$ for a sufficiently large $M$ compared to $N$ and $\delta_0$.

After successfully sample $U \sim \mathcal{H}^{M}$, we can find a specific collision-free outcome $S^{\ast}$ such that $\frac{1}{M}(XX^T) = (UU^T)_{S^{\ast}}$ for our given Gaussian matrix $X$. 
Then, $P(X)$ can also be represented as 
\begin{align}
    P(X) = |\Haf(XX^T)|^2 = M^{N}|\Haf((UU^T)_{S^{\ast}})|^2 = M^{N}\binom{\frac{M}{2} + \frac{N}{2} - 1}{\frac{N}{2}}p_{S^{\ast}}(U),
\end{align}
for the output probability $p_{S^{\ast}}(U)$ of GBS.
Hence, the problem reduces to estimate $p_{S^{\ast}}(U)$ to within additive error $\epsilon_0M^{-N}N!$ with high probability (at least $1 - \frac{3}{4}\delta_0$, as a failure probability to sample $U$ is $\frac{\delta_0}{4}$). 
Here, observe that
\begin{align}
    \frac{N!}{M^N}\geq \frac{1}{2}\frac{N! (M-N)!}{M!} = \frac{1}{2}\binom{M}{N}^{-1} ,
\end{align}
where the inequality holds as long as $M = \omega(N^2)$. 
Therefore, we can further reduce the problem to estimate $p_{S^{\ast}}(U)$ to within additive error $\frac{\epsilon_0}{2}\binom{M}{N}^{-1}$ with probability at least $1 - \frac{3}{4}\delta_0$.

Now, for general $N$-photon collision-free outcome $S$, given $p_{S}(U)$ as the ideal output probability, let $\bar{p}_{S}(U)$ be the output probability distribution of the approximate sampler $\mathcal{S}_0$ with the given unitary circuit $U$. 
Also, let $\mathcal{C}_{M,N}$ be a set of ``collision-free" $N$-photon outcomes of GBS over $M$ modes, and let $\mathcal{G}_{M,N}$ be a uniform distribution over $\mathcal{C}_{M,N}$.
Note that by definition, $\sum_{S \in \mathcal{C}_{M,N}} |\bar{p}_{S}(U) - p_{S}(U)| \le 2\beta$, because total variation distance considers all possible $N$-photon outcomes including collision outcomes.
Then, $\bar{p}_{S}(C)$ satisfies 
\begin{align}\label{averageofdifference}
    \E_{S\sim\mathcal{G}_{M,N}} \left[|\bar{p}_{S}(U) - p_{S}(U)|\right] = \frac{1}{\binom{M}{N}} \sum_{S\in\mathcal{C}_{M,N}}|\bar{p}_{S}(U) - p_{S}(U)|
    \le \frac{2\beta}{\binom{M}{N}}.
\end{align}
Using Eq.~\eqref{averageofdifference} and Markov's inequality, $\bar{p}_{S}(U)$ satisfies
		\begin{align}
			\Pr_{S\sim\mathcal{G}_{M,N}} \left[|\bar{p}_{S}(U) - p_{S}(U)| \geq  \frac{\beta \kappa}{\binom{M}{N}}\right] \leq \frac{\binom{M}{N}}{\beta \kappa} \E_{S\sim\mathcal{G}_{M,N}} \left[|\bar{p}_{S}(U) - p_{S}(U)|\right] \leq  \frac{2}{\kappa},
		\end{align}
for all $\kappa > 2$.
Here, note that the Haar-random unitary matrix $U\sim\mathcal{H}^{M}$ has a symmetric property, such that the distribution of $p_{S}(U)$ is equivalent to the distribution of $p_{S^*}(PU)$ for a fixed outcome $S^*$ and a uniformly random permutation $P$, which is also equivalent to $p_{S^*}(U)$ by the translation invariance of Haar-random unitary $U$. 
Therefore, we equivalently have
\begin{align}
    \Pr_{U\sim\mathcal{H}^{M}} \left[|\bar{p}_{S^{\ast}}(U) - p_{S^*}(U)| \geq  \frac{\beta \kappa}{\binom{M}{N}}\right]  \leq  \frac{2}{\kappa}.
\end{align}

Here, given access to $\mathcal{S}_0$, using Stockmeyer's algorithm~\cite{stockmeyer1985approximation} whose complexity is in $\rm{BPP}^{\rm{NP}}$, obtaining the estimate $\tilde{p}_{S^*}(U)$ of $\bar{p}_{S^*}(U)$ satisfying 
		\begin{equation}
			\Pr\left[|\tilde{p}_{S^*}(U) - \bar{p}_{S^*}(U)| \ge \alpha \bar{p}_{S^*}(U)\right] \le \frac{1}{2^N},
		\end{equation}
in polynomial time in $N$ and $\alpha^{-1}$ is in $\rm{BPP}^{\rm{NP}^{\mathcal{S}_0}}$. 
Here, using $\E_{S \sim \mathcal{G}_{M,N}}[\bar{p}_{S}(U)] = \binom{M}{N}^{-1} \sum_{S\in\mathcal{C}_{M,N}}\bar{p}_{S}(U) \le \binom{M}{N}^{-1}$ and Markov inequality, for all $\lambda > 1$, we have
\begin{align}
    \Pr_{S \sim \mathcal{G}_{M,N}} \left[ \bar{p}_{S}(U) \ge \frac{\lambda}{\binom{M}{N}} \right] \le \frac{1}{\lambda} \quad \Rightarrow \quad \Pr_{U\sim\mathcal{H}^{M}} \left[ \bar{p}_{S^*}(U) \ge \frac{\lambda}{\binom{M}{N}} \right] \le \frac{1}{\lambda},
\end{align}
where the right-hand side comes from the same symmetry used previously.
Then, we arrive at
\begin{align}
    \Pr \left[|\tilde{p}_{S^*}(U) - \bar{p}_{S^*}(U)| \ge \frac{\alpha \lambda}{\binom{M}{N}} \right] &\le \Pr_{U\sim\mathcal{H}^{M}} \left[ \bar{p}_{S^*}(U) \ge \frac{\lambda}{\binom{M}{N}} \right] + \Pr \left[|\tilde{p}_{S^*}(U) - \bar{p}_{S^*}(U)| \ge \alpha \bar{p}_{S^*}(U)\right] \\
    &\le \frac{1}{\lambda} + \frac{1}{2^N}.
\end{align} 
Putting all together, by triangular inequality and union bound, one can probabilistically approximate $p_{S^*}(U)$ by $\tilde{p}_{S^*}(U)$ over $U$ such that
\begin{align}
    \Pr \left[|\tilde{p}_{S^*}(U) - p_{S^*}(U)| \ge \frac{\beta \kappa + \alpha \lambda}{\binom{M}{N}} \right] 
    &\le \Pr_{U\sim\mathcal{H}^{M}} \left[|\bar{p}_{S^{\ast}}(U) - p_{S^*}(U)| \geq  \frac{\beta \kappa}{\binom{M}{N}}\right]  + \Pr \left[|\tilde{p}_{S^*}(U) - \bar{p}_{S^*}(U)| \ge \frac{\alpha\lambda}{\binom{M}{N}} \right] \\
    &\le \frac{2}{\kappa} + \frac{1}{\lambda} + \frac{1}{2^N} ,
\end{align}
within complexity class $\rm{BPP}^{\rm{NP}^{\mathcal{S}_0}}$.
Now, given $\epsilon_0$ and $\delta_0$, we set $\kappa/2 = \lambda = 4/\delta_0$ and $\beta = \alpha = \epsilon_0\delta_0/{24}$. 
Then, $\beta \kappa + \alpha \lambda = 12\beta/\delta_0 = \epsilon_0/2$ and $\frac{2}{\kappa} + \frac{1}{\lambda} + \frac{1}{2^N} = \frac{1}{2}\delta_0 + \frac{1}{2^N} \le \frac{3}{4}\delta_0$.
Therefore, one can estimate $p_{S^{\ast}}(U)$ to within additive error $\frac{\epsilon_0}{2}\binom{M}{N}^{-1}$ with probability at least $1 - \frac{3}{4}\delta_0$ in complexity class $\rm{BPP}^{\rm{NP}^{\mathcal{S}_0}}$.	
This completes the proof. 

\end{proof}

\section{Series expansion of the polynomial $R_{X}(\eta)$}\label{appendix:section:series}

In this section, we establish a series expansion of the polynomial $R_{X}(\eta)$, which we have previously defined as
\begin{align}
    R_{X}(\eta) = \Haf{\begin{pmatrix}
XX^T & (1-\eta)M\tanh{r}\mathbb{I}_{N} \\
(1-\eta)M\tanh{r}\mathbb{I}_{N} & X^{*}X^{\dag}
\end{pmatrix}}    .
\end{align}
Note that, considering the input matrix to hafnian as an adjacency matrix of a graph, the hafnian of the input matrix can be interpreted as the summation of the weights of perfect matchings in the graph. 
Accordingly, we now consider a graph $G = (V,E)$ with $|V| = 2N$ vertices, corresponding to the adjacency matrix $\mathcal{A}$ given by
\begin{align}\label{adjacencymatrix}
   \mathcal{A} = \begin{pmatrix}
XX^T & (1-\eta)M\tanh{r}\mathbb{I}_{N} \\
(1-\eta)M\tanh{r}\mathbb{I}_{N} & X^{*}X^{\dag}
\end{pmatrix}  .
\end{align}
Note that the graph $G$ corresponding to the adjacency matrix $\mathcal{A}$ has a simple bipartite structure, as illustrated in Fig.~\ref{fig:graphrepresentation}. 
Specifically, the vertex set $V$ can be divided into the vertex set $V_1$ and $V_2$, where $i$th vertex in $V_1$ corresponds to the matrix index $i$ for $i \in [N]$ (similarly for $V_2$ with matrix index $N + i$ for $i \in [N]$). 
We can also divide the edge set $E$ with $E_1$, $E_2$, and $E_I$. 
Here, vertices in $V_1$ are connected to each other by edges in $E_1$, such that the adjacency matrix of the subgraph $(V_1,E_1)$ corresponds to $XX^T$ (similarly for the subgraph $(V_2, E_2)$, where now the adjacency matrix is $X^*X^{\dag}$). 
Additionally, the $i$th vertex of $V_1$ is connected with the $i$th vertex of $V_2$ by an edge with its weight $(1-\eta)M\tanh r$, for all $i \in [N]$.
We will denote the set of those $N$ number of edges as $E_I$.

\begin{figure*}[bt]
\includegraphics[width=0.95\linewidth]{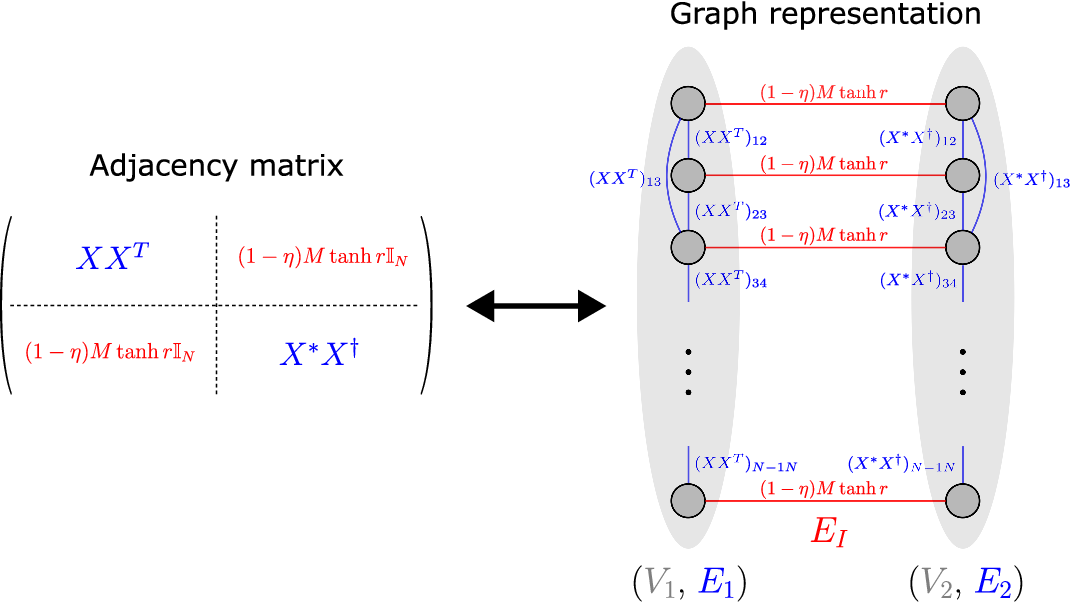}
\caption{Graph representation of the $2N$ by $2N$ adjacency matrix $\mathcal{A}$ described in Eq.~\eqref{adjacencymatrix}.
The corresponding graph $G = (V,E)$ can be represented as a bipartite graph, whose vertices $V$ can be divided into the vertex sets $V_1$ and $V_2$.
Here, for each $i \in [N]$, the $i$th vertex in $V_1$ corresponds to the matrix index $i$ of $\mathcal{A}$, and the $i$th vertex in $V_2$ corresponds to the matrix index $N + i$ of $\mathcal{A}$. 
Also, the edge set $E$ can be divided into the edge sets $E_1$, $E_2$, and $E_I$. 
Here, the edges in $E_1$ and $E_2$, denoted as the blue edges in the figure, connect the vertices inside $V_1$ and $V_2$, respectively. 
Accordingly, the adjacency matrix of the subgraph $(V_1,E_1)$ and $(V_2, E_2)$ is given by $XX^T$ and $X^*X^{\dag}$, respectively. 
Finally, the edges in $E_I$, denoted as the red edges in the figure, connect the vertices between $V_1$ and $V_2$.
Specifically, the $i$th edge in $E_I$ connects the $i$th vertex of $V_1$ with $i$th vertex of $V_2$ with its weight $(1-\eta)M\tanh r$, for all $i \in [N]$.
}
\label{fig:graphrepresentation}
\end{figure*}

Using these conventions, we can divide the case of perfect matchings according to the \textit{number of edges} that belong to $E_I$. 
Since $N$ is even in our case, only an even number of edges in perfect matchings can belong to $E_I$, so let us denote this number as $2n$ for $n \in [N/2]$. 
Also, for each $n$, we can further divide the perfect matchings based on which $2n$ edges are selected from $E_I$. 
For this purpose, let us also denote a set $J \subseteq [N]$ with $|J| = N - 2n$, such that we select the edges in $E_I$ according to the elements of $[N]\setminus J$.
Then, given a specific $J$, the possible perfect matching is determined by selecting perfect matching along the $N-2n$ number of vertices corresponding to $J$ in $V_1$ (and also for $V_2$), and selecting $2n$ number of edges in $E_I$ corresponding to $[N]\setminus J$.

Let $\mu$ and $\nu$ be possible perfect matching permutations along $N-2n$ modes.
Then, given a specific $n$ and $J$, the summation of the weights of all possible perfect matchings can be expressed as  
\begin{align}
\sum_{\mu} \sum_{\nu} (1-\eta)^{2n} M^{2n} \tanh^{2n}{r} \prod_{j=1}^{N/2 - n}[(XX^T)_J]_{\mu(2j-1),\mu(2j)} [(X^*X^{\dag})_J]_{\nu(2j-1),\nu(2j)} \\
= (1-\eta)^{2n} M^{2n} \tanh^{2n}{r} |\Haf((XX^T)_{J})|^2,
\end{align}
where $(XX^T)_{J}$ is a $|J|$ by $|J|$ matrix defined by taking rows and columns of $XX^T$ according to $J$.
Exceptionally, for $n = N/2$ that makes $J$ an empty set, the weight of perfect matching is $(1-\eta)^{N}M^{N}\tanh^{N}{r}$.

Therefore, as $R_{X}(\eta)$ is a summation of the above term for all possible $n$ and $J$, $R_{X}(\eta)$ can finally be represented as
\begin{align}
R_{X}(\eta) = \sum_{n=0}^{N/2 - 1} \sum_{\substack{J \subseteq [N]\\|J|= N -2n}} (1-\eta)^{2n}M^{2n}\tanh^{2n}{r} \left|\Haf\left((XX^{T})_{J}\right)\right|^2
+ (1-\eta)^{N}\tanh^{N}{r}M^N ,
\end{align}
which now becomes a series expansion of $(1-\eta)$. 

\section{Proof of Lemma~\ref{lemma:distancebytrunctionissmall}}\label{appendix: section: proof of lemma}

In the section, we provide a proof of Lemma~\ref{lemma:distancebytrunctionissmall}.
More specifically, we find that for $X \sim \mathcal{N}(0,1)_{\mathbb{C}}^{N \times M}$, $\sinh^2{r} = \frac{N}{\eta M}$, $\eta \ge \eta_{\rm{min}} = 1 - \frac{k_{\rm{max}}}{N + k_{\rm{max}}}$, and $l > k_{\rm{max}}$, $R_{X}^{(l)}(\eta)$ is close to $R_{X}(\eta)$ over a large portion of $X$ such that 
\begin{align}
\Pr_X\left[ |R_{X}(\eta) - R_{X}^{(l)}(\eta)| > \epsilon_1\binom{\frac{M}{2} + \frac{N}{2} - 1}{\frac{N}{2}}N!\right] < \delta,
\end{align}
for $\epsilon_1$ given by
\begin{align}\label{appendix:epsilon1}
    \epsilon_1 = \frac{Nk_{\rm{max}}^{l}}{2\delta l!}.
\end{align}

\begin{proof}[Proof of Lemma~\ref{lemma:distancebytrunctionissmall}]

We begin by introducing the average value of the squared hafnian $|\Haf(XX^T)|^2$ over the random Gaussian matrix $X \sim \mathcal{N}(0,1)_{\mathbb{C}}^{N \times M}$.
Specifically, by the previous analyses in~\cite{ehrenberg2025transition, ehrenberg2025second}, the first moment of the squared hafnian of the Gaussian matrix product is given by
\begin{align}
    \E_{X \sim \mathcal{N}(0,1)_{\mathbb{C}}^{N \times M}}\left[ |\Haf(XX^T)|^2 \right] &=  \frac{(N-1)!!(M+N-2)!!}{(M-2)!!} \\
    &= \frac{N!}{2^{\frac{N}{2}}\left(\frac{N}{2}\right)!}\frac{2^{\frac{M}{2}+\frac{N}{2}-1}\left(\frac{M}{2} + \frac{N}{2} - 1\right)!}{2^{\frac{M}{2}-1}\left(\frac{M}{2} - 1\right)!} \\
    &= \binom{\frac{M}{2} + \frac{N}{2} - 1}{\frac{N}{2}}N! .
\end{align}
Here, we can generalize this result to $|\Haf((XX^T)_{J})|^2$ case for any $J \subseteq [N]$, since $(XX^T)_{J}$ is equivalent to $YY^{T}$ for $Y$ being submatrix of $X$ (by taking rows of $X$ according to $J$), which can also be interpreted as a random Gaussian matrix $Y \sim \mathcal{N}(0,1)_{\mathbb{C}}^{|J| \times M}$ with its size reduced. 
Hence, for any $J \subseteq [N]$ with $|J| = N - 2n$, we have 
\begin{align}
\E_{X}\left[ \left|\Haf\left[(XX^{T})_{J}\right]\right|^2 \right] =\binom{\frac{M}{2} + \frac{N}{2} - n - 1}{\frac{N}{2} - n}(N - 2n)! .
\end{align}
Using the above results, we now have the following inequalities:  
\begin{align}
\frac{\E_X\left[ (1-\eta)^{2n}\tanh^{2n}{r}M^{2n}\sum_{\substack{J \subseteq [N]\\|J|= N -2n}}\left|\Haf\left[(XX^{T})_{J}\right]\right|^2 \right]}{\binom{\frac{M}{2} + \frac{N}{2} - 1}{\frac{N}{2}}N!} 
&\le  \frac{(1-\eta)^{2n}M^{n}N^{n}}{\eta^n}\binom{N}{2n}\frac{\binom{\frac{M}{2} + \frac{N}{2} - n - 1}{\frac{N}{2} - n}(N - 2n)!}{\binom{\frac{M}{2} + \frac{N}{2} - 1}{\frac{N}{2}}N!} \\ 
&= \frac{(1-\eta)^{2n}M^{n}N^{n}}{\eta^n}\frac{\left(\frac{N}{2}\right)!\left(\frac{M}{2} + \frac{N}{2} - n - 1\right)!}{(2n)!\left(\frac{N}{2}-n\right)!\left(\frac{M}{2} + \frac{N}{2} - 1\right)!} \\
&\le \frac{(1-\eta)^{2n}M^{n}N^{n}}{\eta^n} \frac{1}{(2n)!}\frac{\left(\frac{N}{2}\right)^n}{\left( \frac{M}{2} + \frac{N}{2} - n \right)^{n}}   \\
&\le \frac{(1-\eta)^{2n}M^{n}N^{n}}{\eta^n} \frac{\left(\frac{N}{2}\right)^n}{(2n)!\left(\frac{M}{2}\right)^n}\\
&= \frac{(1-\eta)^{2n}N^{2n}}{(2n)!\eta^n} \\
&\le \frac{(1-\eta_{\rm{min}})^{2n}N^{2n}}{{(2n)!\eta_{\rm{min}}}^n} \\
&\leq \frac{k_{\rm{max}}^{2n}}{(2n)!}, \label{asdf}
\end{align}
where we used $\tanh^2{r} = \frac{N}{\eta M + N} \le \frac{N}{\eta M}$ in the first inequality.
Here, note that the right-hand side term in Eq.~\eqref{asdf} monotonically decreases with increasing $n$ given that $n > k_{\rm{max}}/2$. 
Similarly, we also find that 
\begin{align}
\frac{(1-\eta)^{N}\tanh^{N}{r}M^N}{\binom{\frac{M}{2} + \frac{N}{2} - 1}{\frac{N}{2}}N!}  
&\le  \frac{(1-\eta)^{N}M^{\frac{M}{2}}N^{\frac{N}{2}}\left(\frac{M}{2}-1\right)!\left(\frac{N}{2}\right)!}{\eta^{\frac{N}{2}}N!\left(\frac{M}{2} + \frac{N}{2} - 1\right)!} 
\le  \frac{(1-\eta)^{N}M^{\frac{M}{2}}N^{\frac{N}{2}}\left(\frac{N}{2}\right)^{\frac{N}{2}}}{\eta^{\frac{N}{2}}N!\left(\frac{M}{2}\right)^{\frac{N}{2}}} \\
&\le \frac{(1-\eta_{\rm{min}})^{N}N^N}{{\eta_{\rm{min}}}^{\frac{N}{2}}N!} 
\leq \frac{k_{\rm{max}}^{N}}{N!} .
\end{align}

Now, we find the average value of the summation of truncated terms (i.e.,  $R_{X}(\eta) - R_{X}^{(l)}(\eta)$) over $X \sim \mathcal{N}(0,1)_{\mathbb{C}}^{N \times M}$.
Specifically, given $l > k_{\rm{max}}$, the summation terms over $n \geq l/2$ monotonically decreases with increasing $n$, and accordingly, we have
\begin{align}
&\E_X\left[ R_{X}(\eta) - R_{X}^{(l)}(\eta)\right] \nonumber \\
&= \sum_{n=l/2+1}^{N/2 - 1} \E_{X}\left[ (1-\eta)^{2n}M^{2n}\tanh^{2n}{r}  \sum_{\substack{J \subseteq [N]\\|J|= N -2n}} \left|\Haf\left((XX^{T})_{J}\right)\right|^2 \right] + (1-\eta)^{N}\tanh^{N}{r}M^N \\
&\le \sum_{n=l/2+1}^{N/2} \frac{k_{\rm{max}}^{2n}}{(2n)!}{\binom{\frac{M}{2} + \frac{N}{2} - 1}{\frac{N}{2}}N!} \\
&\le  \frac{N}{2}\frac{k_{\rm{max}}^{l}}{l!}{\binom{\frac{M}{2} + \frac{N}{2} - 1}{\frac{N}{2}}N!} .
\end{align}
Finally, combining the above results and Markov's inequality, given the size of $\epsilon_1$ as in Eq.~\eqref{appendix:epsilon1}, we have
\begin{align}
\Pr_X\left[ |R_{X}(\eta) - R_{X}^{(l)}(\eta)| > \epsilon_1\binom{\frac{M}{2} + \frac{N}{2} - 1}{\frac{N}{2}}N!\right] &< \frac{1}{\epsilon_1}\frac{\E_X\left[ R_{X}(\eta) - R_{X}^{(l)}(\eta)\right] }{\binom{\frac{M}{2} + \frac{N}{2} - 1}{\frac{N}{2}}N!} \\
&\le \frac{1}{\epsilon_1} \frac{Nk_{\rm{max}}^{l}}{2l!}\\
&= \delta ,
\end{align}
concluding the proof. 
    
\end{proof}

\section{Examination on the $Q(\eta)$ factor}\label{appendix:section:efficient}

In this section, we deal with the multiplicative factor $Q(\eta)$ that frequently appears throughout our work, which we have previously defined in the main text as
\begin{align}
Q(\eta) = (1-(1-\eta)^2\tanh^2{r})^{\frac{M}{2} + N} {}_2F_1\left(\frac{M + N}{2}, \frac{N + 1}{2}; \frac{1}{2};(1-\eta)^2\tanh^2r\right).
\end{align}
Specifically, in the following, we examine whether $Q(\eta)$ is efficiently computable up to the desired imprecision level, and identify the upper bound of $Q(\eta)$ in terms of $N$ and $\eta$.

\subsection{Efficient computability of $Q(\eta)$}

We first discuss whether $Q(\eta)$ is efficiently computable to within a desired imprecision level, because one might wonder that the infinite sum in the hypergeometric function ${}_2F_1$ requires extensive computational costs to compute to within a desired imprecision level. 
Here, we only consider the hypergeometric function ${}_2F_1$, because the multiplicative prefactor $(1-(1-\eta)^2\tanh^2{r})^{\frac{M}{2} + N}$ at the front is efficiently computable. 
In short, we find that even a brute-force computation of the polynomial number of summation terms in ${}_2F_1$ is sufficient for the efficient computation of ${}_2F_1$ to within the desired imprecision level, i.e., $e^{-\text{poly}(N)}$ for any desired $\text{poly}(N)$.

To examine whether the hypergeometric function is efficiently computable, we investigate the size of the summation term on the right-hand side of Eq.~\eqref{appendix:hypergeometric} when $n$ is sufficiently larger than $M$. 
Specifically, when $n = \text{poly}(N) \gg M$, observe that 
\begin{align}
\frac{(\frac{M + N}{2})_n (\frac{N + 1}{2})_n}{(\frac{1}{2})_n n!}  &= \frac{(M+N+2n-2)!!(N+2n-1)!!}{(2n)!(M+N-2)!!(N-1)!!} \\
&\le \frac{2^{2n}(\frac{N}{2}+n)!(\frac{M+N}{2}+n)!}{(2n)!(\frac{N}{2})!(\frac{M+N}{2})!} \\
&< n\binom{n+\frac{N}{2}}{n}\binom{n+\frac{M+N}{2}}{n} \label{inequality:centralbinomialcoefficient} \\
&\le n\frac{(n+\frac{N}{2})^{\frac{N}{2}}}{(\frac{N}{2})!}\frac{(n+\frac{M+N}{2})^{\frac{M+N}{2}}}{(\frac{M+N}{2})!} \\
&< 2^n ,
\end{align}
where we used the lower bound of the central binomial coefficient at Eq.~\eqref{inequality:centralbinomialcoefficient}. 

Based on this understanding, we can estimate the hypergeometric function by computing the polynomial number of the summation terms on the right-hand side of Eq.~\eqref{appendix:hypergeometric}, from $n = 0$ to $n = m = \text{poly}(N)$ for sufficiently large $m$.
Then, the additive imprecision is upper bounded by 
\begin{align}
    \sum_{n = m+1}^{\infty}\frac{(\frac{M + N}{2})_n (\frac{N + 1}{2})_n}{(\frac{1}{2})_n n!} (1-\eta)^{2n}\tanh^{2n}{r} &<  \sum_{n = m+1}^{\infty}\left( \frac{2(1-\eta)^2N}{\eta M} \right)^n \\
    &< \left( \frac{2(1-\eta)^2N}{\eta M} \right)^{m+1}\left( 1 - \frac{2(1-\eta)^2N}{\eta M} \right)^{-1} ,\label{appendix:upperboundoferrorforhypergeometric}
\end{align}
where we used the fact that $\tanh^2{r} = \frac{N}{\eta M + N} < \frac{N}{\eta M}$. 
Note that the right-hand side of Eq.~\eqref{appendix:upperboundoferrorforhypergeometric} scales $e^{-\text{poly}(N)}$ which can be made arbitrarily small by setting $m$ sufficiently large polynomial. 
Therefore, $Q(\eta)$ is efficiently computable to within additive imprecision $e^{-\text{poly}(N)}$ for arbitrarily large $\text{poly}(N)$.

\subsection{Upper bound of $Q(\eta)$}

Next, we present the upper bound of $Q(\eta)$ to obtain the maximal value $Q_{\rm{max}}$ defined in Eq.~\eqref{eq: def of Qmax}.
One can easily find that $(1-(1-\eta)^2\tanh^2{r})^{\frac{M}{2} + N} = O\left(e^{-\frac{N}{2}\frac{(1-\eta)^2}{\eta}} \right)$.
Also, since the post-selection probability given in Eq.~\eqref{appendix:postselectionprobability} is smaller than unity, one simple way to find the upper bound of the hypergeometric function is
\begin{align}
    {}_2F_1 \left(\frac{M + N}{2}, \frac{N + 1}{2}; \frac{1}{2};(1-\eta)^2\tanh^2r\right) &\le \frac{\cosh^{M}r}{\eta^{N}\tanh^{N}r}\binom{\frac{M}{2} + \frac{N}{2} -1 }{\frac{N}{2}}^{-1} \\
    &\le O\left(\sqrt{N}e^{\frac{N}{2}\left(\frac{1}{\eta}-1-\log\eta\right)}\right), \label{qdf}
\end{align}
where in Eq.~\eqref{qdf} we directly used the results from Eq.~\eqref{tyty} to Eq.~\eqref{qwerty}. 
Hence, $Q(\eta)$ can be upper bounded as 
\begin{align}\label{dfasdfqw}
    Q(\eta) \le O\left(\sqrt{N}e^{\frac{N}{2}\left(\frac{1}{\eta}-1-\log\eta - \frac{(1-\eta)^2}{\eta}\right)}\right) .
\end{align}
Especially, in the low-loss regime such that $(1-\eta) \leq O(\frac{1}{\sqrt{N}})$, we have $\frac{N}{2}\left(\frac{1}{\eta}-1-\log\eta \right) = (1-\eta)N + O(1)$ from Eq.~\eqref{qqwer}, and finally the right-hand side term in Eq.~\eqref{dfasdfqw} takes the simplified form as
\begin{align}\label{appendix: bound of Q}
    Q(\eta) \le O\left(\sqrt{N}e^{(1-\eta)N}\right) .
\end{align}

We remark that there would exist a tighter bound than Eq.~\eqref{appendix: bound of Q} by an multiplicative factor $\frac{1}{\sqrt{N}}$, because we expect that the post-selection probability $\Pr[N]$ to obtain the mean photon number $N = \eta M \sinh^2{r}$ scales $\Theta(\frac{1}{\sqrt{N}})$ similarly to the ideal case without loss.
More precisely, even though we have $\Pr[N = \eta M \sinh^2{r}] \ge \Omega(\frac{1}{\sqrt{N}})$ as given in Eq.~\eqref{appendix:lowerboundofpostselectionprobability}, we also expect that $\Pr[N = \eta M \sinh^2{r}] \le O(\frac{1}{\sqrt{N}})$ for $\eta < 1$, which would lead to a tighter bound $Q(\eta) \le O\left(e^{(1-\eta)N}\right)$. 
Still, the bound in Eq.~\eqref{appendix: bound of Q} is sufficient for our hardness analysis.

\section{Total variation distance bound between the ideal and lossy GBS}\label{appendix: section: TVD between lossy GBS and ideal GBS}

Let $\rho_{\text{id}}$ be the output state of the ideal GBS (i.e., before measuring on photon number basis), and let $\rho_{\text{n}}$ be the output state of the lossy GBS after the beam splitter loss characterized by $\eta$. 
Recall that, for any measurement basis, the total variation distance between two quantum states $\rho_{\text{id}}$ and $\rho_{\text{n}}$ can be upper bounded by using the quantum fidelity~\cite{fuchs1999cryptographic} such that
\begin{align}
    \mathcal{D}(\rho_{\text{n}}, \rho_{\text{id}}) \le \sqrt{1 - F(\rho_{\text{n}}, \rho_{\text{id}})} = \sqrt{1 - F(\rho_{\text{n, in}}, \rho_{\text{id, in}})} ,
\end{align}
where we denote $\rho_{\text{id, in}}$ as the input state of the ideal GBS before going through the unitary circuit of GBS, which is a product of $M$ squeezed vacuum states in our setup (and accordingly a pure state).
Also, we denote $\rho_{\text{n, in}}$ as the input state of the lossy GBS where all photon loss is pushed onto the input (note that beam splitter loss channel commutes with linear optical networks), i.e., a product of $M$ squeezed vacuum states followed by the beam splitter loss. 
Note that, the covariance matrices (in $xp$ basis) of $\rho_{\text{id, in}}$ and $\rho_{\text{n, in}}$ are given by $\sigma_{0}$ and $\sigma_{in}$, given in Eq.~\eqref{eq: input xp covariance matrix before loss} and Eq.~\eqref{eq: input xp covariance matrix after loss}, respectively.

To proceed, quantum fidelity between two Gaussian states $\rho_{\text{n, in}}$ and $\rho_{\text{id, in}}$ with covariance matrices $\sigma_{0}$ and $\sigma_{in}$ and zero means, one of which is pure, can be written as~\cite{spedalieri2012limit}
\begin{align}\label{eq: fidelity between fin and id}
   F(\rho_{\text{n, in}}, \rho_{\text{id, in}}) = \frac{1}{\sqrt{| \sigma_{in} +\sigma_{0}|}} .
\end{align}
Here, using Eq.~\eqref{eq: input xp covariance matrix before loss} and Eq.~\eqref{eq: input xp covariance matrix after loss}, observe that
\begin{align}
    | \sigma_{in} +\sigma_{0}| &= \left\vert \begin{pmatrix}
        \frac{e^{2r}}{2} & 0 \\ 0 & \frac{e^{-2r}}{2}
    \end{pmatrix}  +  \begin{pmatrix}
\frac{\eta e^{2r} + (1-\eta)}{2}  & 0 \\
0 & \frac{\eta e^{-2r} + (1-\eta)}{2}
\end{pmatrix}\right\vert^{M} \\
    &= \left\vert\begin{pmatrix}
        e^{2r} & 0 \\ 0 & e^{-2r}
    \end{pmatrix}  +  \frac{1-\eta}{2} \begin{pmatrix}
    1 - e^{2r}  & 0 \\
    0 & 1 - e^{-2r}
\end{pmatrix}\right\vert^{M} \\
    &= \left\vert\begin{pmatrix}
        1 & 0 \\ 0 & 1
    \end{pmatrix}  +  \frac{1-\eta}{2} \begin{pmatrix}
    e^{-2r} - 1  & 0 \\
    0 & e^{2r} - 1
\end{pmatrix}\right\vert^{M} \\
    &\leq \left( 1 + \frac{1 - \eta}{4} \Tr\left[  \begin{pmatrix}
    e^{-2r} - 1  & 0 \\
    0 & e^{2r} - 1
\end{pmatrix} \right]\right)^{2M} \label{eq: AM-GM} \\
    &= \left( 1 + (1-\eta)\sinh^2 r \right)^{2M}, 
\end{align}
where we used the AM–GM inequality in Eq.~\eqref{eq: AM-GM}.
Therefore, we have a lower bound for the fidelity by
\begin{align}
    F(\rho_{\text{n, in}}, \rho_{\text{id, in}}) \geq \left( 1 + (1-\eta)\sinh^2 r \right)^{-M} \geq 1 - (1-\eta)M\sinh^2 r , 
\end{align}
where we used the fact that $(1+x/M)^{M} \le (1-x)^{-1}$ for $0\le x < 1$, and also assumed that $(1-\eta)M\sinh^2 r < 1$.

To sum up, the total variation distance is bound by $\beta < 1$ when $\eta$ satisfies
\begin{align}
    \mathcal{D}(\rho_{\text{n}}, \rho_{\text{id}}) \le \sqrt{1 - F(\rho_{\text{n, in}}, \rho_{\text{id, in}})}  \leq \sqrt{(1-\eta)M\sinh^2 r} \leq \beta ,
\end{align}
which gives the condition for $\eta$ that
\begin{align}
     (1-\eta)M\sinh^2 r \leq \beta^2 .
\end{align}
This bound satisfies our assumption $(1-\eta)M\sinh^2 r < 1$ given that $\beta < 1$, thereby concluding the proof.

\end{document}